\documentclass[%
11pt,
reprint,
onecolumn,
tightenlines,
superscriptaddress,
preprintnumbers,
nofootinbib,
amsmath,amssymb,amsthm,
physrev,
eqsecnum,tikz,
]{revtex4-2}

\usepackage{isomath}
\usepackage{amsbsy}
\usepackage{amssymb}
\usepackage{amscd}
\usepackage{amsfonts}
\usepackage{stmaryrd}
\usepackage{euscript}
\usepackage[utf8]{inputenc}
\usepackage[T1]{fontenc}
\usepackage{newtxtext} 
\everymath{\displaystyle}
\usepackage{exscale}
\usepackage{microtype}
\usepackage{booktabs}
\usepackage{algorithm}
\usepackage{algpseudocode}
\usepackage{hyperref}

\usepackage{amsmath,amsthm}
\usepackage{mathtools}          
\usepackage{multirow}
\usepackage{array}
\usepackage{cleveref}

\usepackage{graphicx}
\usepackage{boxedminipage}
\usepackage{calc}
\usepackage[dvipsnames]{xcolor}
\graphicspath{ {media/} }
\usepackage[caption=false,justification=raggedright]{subfig}

\usepackage{setspace}
\usepackage{enumitem}
\setitemize{noitemsep,topsep=0pt,parsep=0pt,partopsep=0pt}
\setenumerate{noitemsep,topsep=0pt,parsep=0pt,partopsep=0pt}
\setdescription{noitemsep,topsep=0pt,parsep=0pt,partopsep=0pt}

\usepackage[normalem]{ulem}

\usepackage{orcidlink}
\usepackage{siunitx}
\usepackage[small]{titlesec}

\titlespacing*{\section}{0pt}{12pt plus 4pt minus 2pt}{2pt plus 2pt minus 2pt}
\titlespacing*{\subsection}{0pt}{12pt plus 4pt minus 2pt}{2pt plus 2pt minus 2pt}
\titlespacing*{\subsubsection}{0pt}{12pt plus 4pt minus 2pt}{2pt plus 2pt minus 2pt}
\titlespacing*{\paragraph}{0pt}{12pt plus 4pt minus 2pt}{2pt plus 2pt minus 2pt}

\makeatletter
    \renewcommand*{\thesection}{\arabic{section}}
    \renewcommand*{\thesubsection}{\thesection.\Alph{subsection}}
    \renewcommand*{\p@subsection}{}
    \renewcommand*{\thesubsubsection}{\thesubsection.\arabic{subsubsection}}
    \renewcommand*{\p@subsubsection}{}
\makeatother

\usepackage{upgreek}

\newcommand{\bw}{\mathbold{w}}           
\newcommand{\bv}{\mathbold{v}}           
\newcommand{\bx}{\mathbold{x}}           
\newcommand{\bu}{\mathbold{u}}           
\newcommand{\bxi}{\mathbold{\xi}}        
\newcommand{\bp}{\mathbold{p}}           
\newcommand{\be}{\mathbold{e}}           
\newcommand{\bA}{\mathbold{A}}           
\newcommand{\bL}{\mathbold{L}}           
\newcommand{\bF}{\mathbold{F}}           
\newcommand{\bM}{\mathbold{M}}           
\newcommand{\bI}{\mathbold{I}}           
\newcommand{\bD}{\mathbold{D}}           
\newcommand{\bW}{\mathbold{W}}           
\newcommand{\bsigma}{\mathbold{\Sigma}}           
\newcommand{\bC}{\mathbold{C}}           
\newcommand{\bV}{\mathbold{V}}           
\newcommand{\gf}{g}                
\newcommand{\ff}{f}                
\newcommand{\coll}{\mathcal{Q}}    
\newcommand{\T}{\mathsf{T}}        

\newtheorem{proposition}{Proposition}[section]

\theoremstyle{definition}

\AtEndEnvironment{definition}{\null\hfill\qedsymbol}

\AtEndEnvironment{remark}{\null\hfill\qedsymbol}

\AtEndEnvironment{example}{\null\hfill\qedsymbol}

\AtEndEnvironment{assumption}{\null\hfill\qedsymbol}

\newcommand{\bfalpha}{\mathbold {\alpha}}

\newcommand{\bfzero}{\mathbf{0}}

\DeclareMathOperator{\trace}{tr}

\newcommand{\parderiv}[2]{\frac{\partial #1}{\partial #2}}
\newcommand{\dm}{\ \mathrm{d}}
\newcommand{\deriv}[2]{\frac{\dm #1}{\dm #2}}

\newcommand{\bft}{{\mathbold t}}

\newcommand{\bfv}{{\mathbold v}}
\newcommand{\bfw}{{\mathbold w}}
\newcommand{\bfx}{{\mathbold x}}

\newcommand{\bfA}{{\mathbold A}}

\newcommand{\bfC}{{\mathbold C}}
\newcommand{\bfD}{{\mathbold D}}

\newcommand{\bfF}{{\mathbold F}}

\newcommand{\bfI}{{\mathbold I}}

\newcommand{\bfL}{{\mathbold L}}

\newcommand{\bfV}{{\mathbold V}}
\newcommand{\bfW}{{\mathbold W}}

\begin{document}


\preprint{To appear in Physics of Fluids (DOI: \href{https://doi.org/10.1063/5.0346853}{10.1063/5.0346853})}

\title{Anisotropic Thermalization in Far-from-Equilibrium Flows}

\author{Arnab Debnath \orcidlink{0009-0008-0829-3385}} 
    \affiliation{Program in Computational Mechanics, Carnegie Mellon University}
    \affiliation{Department of Civil and Environmental Engineering, Carnegie Mellon University}

\author{Timothy Breitzman}
    \affiliation{Materials and Manufacturing Directorate, Air Force Research Laboratory}

\author{Kaushik Dayal \orcidlink{0000-0002-0516-3066}}
    \email{Kaushik.Dayal@cmu.edu}
    \affiliation{Department of Civil and Environmental Engineering, Carnegie Mellon University}
    \affiliation{Center for Nonlinear Analysis, Department of Mathematical Sciences, Carnegie Mellon University}
    \affiliation{Department of Mechanical Engineering, Carnegie Mellon University}

\date{\today}


\begin{abstract}
We present a deterministic discontinuous Galerkin (DG) finite-element solution of the Boltzmann equation, without moment-closure approximations, under a class of far-from-equilibrium deformations.
Specifically, we consider affine flows which reduce the Boltzmann equation to a purely velocity-space problem for the reduced distribution function in a reduced velocity field.
We solve the reduced equation using a tensor-product Lagrange DG discretization for four representative flows: simple shear, pressure shear,
bi-directional shear, and a vortex flow.
Our principal finding is that the velocity distribution is well-approximated by an anisotropic Gaussian throughout the evolution, despite the non-equilibrium conditions.  
Further, we show the evolution of the covariance tensor of the Gaussian distribution is equal to the inverse of the right Cauchy-Green tensor in the free-streaming limit without collisions.
This prediction compares very well with the numerical solution at short times; at longer times, they grow apart, reflecting the influence of particle collisions.
\end{abstract}

\maketitle


\section{Introduction}
\label{sec:intro}

The Boltzmann equation provides a description of matter in the regime that is between the completely discrete atomistic and the classical continuum models \cite{truesdell1980}. It provides the fundamental kinetic-theory basis for non-equilibrium
thermodynamics, rarefied gas flow \cite{truesdell1980}, and the derivation of macroscopic
transport equations from first principles
\cite{cercignani1988,chapman1970,bird1994}.
It has been applied to model dilute systems of large numbers of particles that interact in a pairwise manner, for instance phonons carrying thermal energy that lead to heat transfer \cite{ziman1960electrons}; electrons in semiconductors \cite{Deng2023SiCBTE}; plasmas \cite{Petrov2024ElectronBoltzmannHypersonicPlasmas}; and particulate  materials \cite{Jenkins1983RapidFlowGranular,Syamlal1993MFIXTheoryGuide,Carrillo2021KineticGranularMaterials}.

The Boltzmann equation is posed in terms of a molecular distribution function $f(\bfx,\bfv,t)$, that provides the number of particles at spatial location $\bfx$ with velocity $\bfv$, at time $t$. 
Direct numerical solution is challenging: the distribution function depends on three spatial coordinates, three velocity coordinates, and time, making standard numerical approaches prohibitively expensive.

\paragraph*{Affine flows and the homoenergetic reduction.}

An important class of exact solutions of the Boltzmann equation corresponds to spatially affine (homogeneous) velocity fields.
The idea that molecular dynamics simulations can be performed for flows of this type, using only a small symmetry-reduced unit cell, originates in the concept of objective molecular dynamics (OMD) introduced by James and collaborators \cite{james2006,dumica2007}.
The time-dependent non-equilibrium generalization, appropriate for transient shear and viscometric flows, was developed by Dayal and James for molecular dynamics \cite{dayal2010} and shown to hold also for the Boltzmann equation \cite{dayal2012,pahlani2023a,pahlani2023b,pahlani2023c}.
Specifically, there exists a large class of flows that reduce the molecular distribution function, denoted the \emph{homoenergetic reduction}, that have form $f(\bfx,\bfv,t)=g(\bfw,t)$, where $\bfw=\bfv - \bfA (\bfI + \bfA t)^{-1} \bfx$ is a change of variables with $\bfA$ an arbitrary constant tensor.
This reduction is motivated by fundamental invariances related to frame-indifference.

These solutions are generically nonequilibrium flows that evolve in time, and represent a vast variety of physical flows for different choices of $\bfA$.
These solutions are related to, and can be considered a generalization of, the Lees-Edwards boundary conditions that are used to simulate Couette flows with molecular dynamics \cite{lees1972,Clausen2011CapsuleSuspensions,Brilliantov2004KineticTheoryGranularGases,Gallier2014RoughFrictionalSuspensions,rosenbaum2019effects,rosenbaum2019surfactant}.
The mathematical structure of these \emph{homoenergetic} solutions was
analyzed in depth by James, Nota, and Vel\'azquez
\cite{james2019b,james2020}, who established existence,
self-similar long-time asymptotics, and the precise roles of
collision-dominated and hyperbolic-dominated regimes.
The kinetic-theory description connecting OMD to the Boltzmann equation
was formalized by James, Qi, and Wang \cite{james2024}.

\paragraph*{Numerical methods for the Boltzmann equation.}
Numerical methods for the solution of the Boltzmann equation follow two principal strategies for the velocity-space discretization: (i) spectral methods based on the Fourier transform of the collision operator \cite{pareschi2000,mouhot2006fast,gamba2017}, and (ii) discontinuous Galerkin (DG) finite-element methods in velocity space, pioneered by Alekseenko and Josyula \cite{alekseenko2014} and extended by Jaiswal, Alexeenko, and Hu \cite{jaiswal2019a,jaiswal2019b} and by Zhang and Gamba \cite{zhang2018}.
Atomistic OMD simulations \cite{pahlani2023a,pahlani2023b} provide molecular dynamics trajectories but not a direct solution of the kinetic equation.

\paragraph*{Observation of anisotropic thermalization.}

The relaxation of anisotropic velocity distributions toward
equilibrium is a fundamental problem in non-equilibrium statistical
mechanics. Cross-dimensional relaxation measurements have been
used to infer collision rates: in early precision measurements of
ultracold Cs--Cs elastic scattering, Monroe et al.\ inferred the
scattering rate from the decay of anisotropy in an initially
non-equilibrium energy distribution \cite{monroe1993}. When the
scattering cross-section itself is anisotropic, the thermalization
rate can become direction-dependent, a phenomenon often referred to
as anisotropic thermalization: in ultracold dipolar gases,
the re-thermalization rate can vary by as much as a factor of two with
the orientation of the dipoles relative to the excitation geometry
\cite{bohnjin2014,aikawa2014,wang2021anisotropic}.

For the dilute classical gas, the relaxation behaviour depends
sensitively on the collision kernel. The Bobylev--Krook--Wu solution
\cite{bobylev1984exact} gives a closed form for spatially homogeneous
relaxation of Maxwell molecules and is widely used to verify
deterministic Boltzmann solvers. For homoenergetic flows of the same
affine form considered here, James, Nota and Vel\'azquez
\cite{james2019b} proved that, for Maxwell molecules and a large
class of flows, the long-time dynamics can be governed
by non-Maxwellian self-similar profiles describing far-from-equilibrium
states. For uniform shear flow, analytical and numerical studies of
Maxwell molecules likewise show non-Maxwellian self-similar profiles
and provide evidence for highly anisotropic, algebraically decaying
high-velocity tails \cite{acedo2002tail,duan2021}; a comprehensive
account of the nonlinear-transport properties of sheared gases is
given in the monograph of Garz\'o and Santos \cite{garzo2003}.

A natural way to encode the resulting anisotropy is through a
tensorial temperature, defined from the full kinetic-energy
tensor rather than from its trace alone. In non-equilibrium molecular
dynamics, related ideas appear in modified thermostatting strategies
for imposed flows, in which different Cartesian velocity components
may be thermostatted separately to remove viscous heat without
directly damping the imposed streaming motion; related thermostatting
issues arise in sheared-liquid MD near solid boundaries
\cite{thompson1990} and in later studies of nanoscale slip and
interfacial friction in confined geometries \cite{falk2010}. The same
kinetic-tensor structure underlies the ellipsoidal-statistical BGK
model \cite{holway1966,andries2000}, in which the isotropic
Maxwellian relaxation target is replaced by an ellipsoidal Gaussian
whose temperature tensor is tied to the stress tensor, and the
anisotropic moment closures developed for relativistic dissipative
fluid dynamics \cite{molnar2016anisotropic}. Sega and Jedlovszky
\cite{sega2018tensorial} have emphasized that imposing a tensorial
temperature through a thermostat can yield ensembles that differ from
the canonical one for molecular fluids near interfaces, so the concept
requires care when used as a modeling input.

The detailed analytical results just summarised concern primarily
Maxwell molecules, and the tensorial-temperature constructions above
arise as modeling or thermostatting choices. The present work is
complementary in two respects: it addresses hard spheres under the
three-dimensional affine flows considered here via deterministic
discontinuous Galerkin solution of the Boltzmann equation, and the
fitted covariance \(\bsigma(t)\) plays the role of a tensorial
temperature measured from the solution rather than imposed. Its
principal axes and eigenvalue spectrum, analyzed in the remainder of
this section, provide an intrinsic characterization of the anisotropic
thermalization of the gas under sustained affine deformation.

\paragraph*{Contributions of this paper.}
In this work, we examine the time evolution of the invariance-reduced solutions using a Discontinuous Galerkin Finite Element Method (DG-FEM) formulation that has been shown to be well-suited to solving the Boltzmann equation in other settings \cite{alekseenko2014,zhang2018,jaiswal2019a,jaiswal2019b,Barth2006Boltzmann}.
The DG-FEM approach reduces the computation of the Boltzmann collision operator from $O(n^8)$ for a brute-force evaluation to $O(n^5)$, where $n$ is the number of quadrature points per direction in velocity space \cite{alekseenko2014}.
Crucially, in combination with the reduction enabled by the homoenergetic invariance, DG-FEM provides a method to numerically compute solutions of the full Boltzmann collision integral without resorting to moment closure approximations such as BGK \cite{bhatnagar1954} or Grad's thirteen-moment system \cite{levermore1996moment}.
The conservation properties of the collision operator are enforced via a constrained least-squares correction following Zhang and Gamba \cite{zhang2018}, which preserves mass exactly and allows momentum and energy to evolve freely under the applied deformation.
We solve using a tensor-product Lagrange DG discretisation in velocity space on $3^3$ and $5^3$ element meshes ($27$ and $125$ elements), with a fifth-order Adams--Bashforth time integrator.

We apply this formulation to four representative homoenergetic flows: simple shear (Couette flow), pressure shear, bi-directional shear, and a vortex-like flow.
We find that the velocity distribution remains well-approximated by an anisotropic Gaussian throughout the evolution, starting from an isotropic equilibrium Maxwellian and reaching anisotropy ratios $\lambda_1/\lambda_3$ of up to $9.96$. 
We further derive that in the free-streaming regime, i.e., when collisions are neglected, the covariance tensor satisfies
\begin{equation}
  \bsigma(t) = (\bF^\T\bF)^{-1},
  \qquad \bF(t) = \bI + t\bA,
  \label{eq:Cpred_intro}
\end{equation}
where $\bfF$ is the deformation gradient tensor and $\bfC := \bfF^\T \bfF$ is the right Cauchy-Green tensor.
The kinematic prediction~\eqref{eq:Cpred_intro} is compared against the numerical results in various relevant velocity-space planes via a per-plane and we find close agreement for short times while the difference grows over time and is related to the continuum spin tensor.

The paper is organized as follows.
\Cref{sec:formulation} presents the homoenergetic reduction and the
four test flows.
\Cref{sec:numerics} describes the DG discretization, collision operator,
conservation routine, and time integration.
\Cref{sec:kinematic} derives the free-streaming kinematic prediction.
\Cref{sec:results} presents the numerical results.
\Cref{sec:discussion} interprets the collision-retardation and summarizes conclusions.

\subsection{Notation}
\label{sec:notation}

\begin{tabbing}
  \hspace{3.5cm} \= \kill
  $(\cdot)_{ij}$          \> $(i,j)$ component of a tensor \\
  $\operatorname{tr}(\cdot)$ \> Trace of a tensor \\
  $|\cdot|$               \> Euclidean norm of a vector \\
  $\det(\cdot)$           \> Determinant of a tensor \\
  $(\circ)$               \> Hadamard product \\
  Repeated indices        \> Einstein summation convention unless stated otherwise \\
  Superscript $(\cdot)'$  \> Post-collision quantity \\
  Superscript $(\cdot)_*$ \> Quantity associated with the second particle \\
\end{tabbing}

\section{Invariance-Reduced Formulation of the Boltzmann Equation}
\label{sec:formulation}

The Boltzmann equation is a well-established model, e.g. \cite{truesdell1980,cercignani1975,Santos2011BGKGranularRoughSpheres,ashcroft1976solid,landau1981physical,cercignani1988}; we summarize it below to fix the notation and assumptions.
The central quantity is the molecular distribution function $f(\bfx,\bfv,t)$ with position $\bfx$, velocity $\bfv$, and time $t$.
It describes the probability $f\dm\bft \dm\bfx \dm\bfv$ of finding a particle in the spatial interval $[\bfx,\bfx+\dm\bfx]$, in the time interval $[t,t+\dm t]$, with velocity in the interval $[\bfv,\bfv+\dm\bfv]$.
The spatially inhomogeneous Boltzmann equation in the absence of any external forces, for a monatomic system
reads
\begin{equation}
  \frac{\partial \ff}{\partial t} + \bv \cdot \nabla_{\bx} \ff
  = \coll[f,f]
  \label{eq:bte_full}
\end{equation}
The collision operator $\coll$ on the right has the form:
\begin{equation}
    \coll[f,f]=\int_{\mathbb{R}^3} \int_{\mathcal{S}} \left( f'_{*} f' - f_* f\right) \mathbb{S} \dm\mathcal{S} \dm\bfv_* = \int_{\mathbb{R}^3} \int_{\mathcal{S}} \left( f'_{*} f' - f_* f\right) |\bv-\bv_*| \dm\mathcal{S} \dm\bfv_*
    \label{eq:collision}
\end{equation}
The pre-collision velocities are denoted as $\bfv$ and $\bfv_*$, and the respective post-collision velocities are denoted using primes: $\bfv'$ and $\bfv'_*$.
For brevity, we write $f(\bfx,\bfv_*,t) \equiv f_*$ and similarly for the primed quantities.
The scattering factor $\mathbb{S}$ represents the physics of the collisional interaction and depends on the nature of the particle interactions.
The scattering domain is represented by $\mathcal{S}$.
An important assumption in this form of the collision operator is that particles only interact in a pairwise manner, i.e., the collisions or interactions always involve only two particles.

Evaluating the high-dimensional collision integral can be computationally prohibitive.
In this work, we focus on hard spheres with elastic interactions because, first, it is relevant to a number of applications in mechanics such as granular flows and rarefied gases; and, second, closed-form expressions for the collision behavior are readily available, e.g. \cite{truesdell1980}.

\subsection{Invariance Reduction of the Boltzmann Equation.}
The invariance reduction used in this paper was motivated by fundamental invariances inherent to molecular dynamics \cite{james2006,dumica2007} and that are inherited by the Boltzmann equation \cite{dayal2012,dayal2010,nota-1,james2020,bobylev2020,haines2009three}.
Consider a periodic molecular dynamics calculation where the positions of the simulated atoms in the zero-th unit cell are $\bx_{0,k}(t)$ where $k$ denotes the atom number in the zero-th unit cell. 
The periodicity is described by the lattice vectors $\mathbf{f}_1(t), \mathbf{f}_2(t), \mathbf{f}_3(t)$; we emphasize that these vectors are allowed to depend on time.
The positions of the image atoms are then given by the expression:
\begin{equation}
  \bx_{\nu,k}(t)
  = \bx_{0,k}(t)
    + \nu^1 \mathbf{f}_1(t)
    + \nu^2 \mathbf{f}_2(t)
    + \nu^3 \mathbf{f}_3(t),
  \quad \nu^1,\nu^2,\nu^3\in\mathbb{Z}.
  \label{eq:image-atoms}
\end{equation}
Requiring that the image trajectories satisfy Newton's equations with no spurious force constrains the lattice vectors to depend at most affinely on time \cite{dayal2010,dayal2012}: $\mathbf{f}_i(t) = (\bI + t\bA)\mathbf{f}_i(0)$, where $\bfA$ is a constant tensor.
Different choices of $\bA$ realize a wide family of viscometric flows \cite{dayal2012}.
For example,
\begin{equation}
  \bA = \begin{pmatrix}0 & \dot\gamma & 0 \\ 0 & 0 & 0 \\ 0 & 0 & 0
        \end{pmatrix}
\end{equation}
gives Couette flow with shear rate $\dot\gamma$, and is equivalent to the Lees--Edwards boundary condition used in non-equilibrium molecular dynamics \cite{lees1972,dayal2010,dayal2012}. 

Moving from the discrete molecular description to the continuum Boltzmann description where each point in space is assumed to contain a large number of molecules, and also converting from the Lagrangian description to the Eulerian, the ansatz for the continuum deformation in terms of the continuum velocity field $\bfV(\bfx,t)$ is:
\begin{equation}
	\bfV(\bfx,t) = \bfA (\bfI + \bfA t)^{-1} \bfx
    \label{eq:L_DayalJames}
\end{equation}
The corresponding spatial velocity gradient is $\bL(t) = \bA(\bI + t\bA)^{-1}$, the deformation gradient is $\bF(t) = \bI + t\bA$, and $\dot{\bF}=\bfL\bfF=\bfA$.
We only consider $t$ such that $(\bI+t\bA)$ is non-singular.

Translating this to the Boltzmann equation, for this class of flows, the velocity statistics of the molecular distribution at any spatial point $\bx$ can be related to the statistics at the origin: 
the probability $\ff(\bfx,\bv,t)\,\mathrm{d}\bx\,\mathrm{d}\bv$ of finding a particle in $[\bx,\bx+\mathrm{d}\bx]$ with velocity in $[\bv,\bv+\mathrm{d}\bv]$ is the same as the probability $\ff(\bfzero,\bv-\bL(t)\bx,t)\,\mathrm{d}\bx\,\mathrm{d}\bv$ of finding a particle at the spatial region $[{\bf 0},{\bf 0}+\dm\bfx]$ with velocity in the interval $[\bfv- \bL\bfx,\bfv+\dm\bfv- \bL\bfx]$.
Physically, for this class of flows, the velocity statistics of the molecular distribution at any point $\bfx$ can be related to the velocity statistics at the origin ${\bf 0}$; if we compute the statistics of the molecular distribution at the origin, we can find the statistics at any point in space. 

This motivates the homoenergetic ansatz
\begin{equation}
  \ff(\bx,\bv,t) = \gf(\bw,t),
  \qquad \bw := \bv - \bL(t)\bx,
  \label{eq:ansatz}
\end{equation}
where $\bw$ is the \emph{reduced velocity}.
While motivated by molecular considerations, the reduction can be verified by direct substitution \cite{dayal2010,dayal2012}.
Substituting~\eqref{eq:ansatz} into~\eqref{eq:bte_full}, we get:
\begin{equation}
    \frac{\partial \ff}{\partial t} + \bv\cdot\nabla_{\bx} \ff
    = 
    \frac{\partial \gf}{\partial t}
    - (\dot{\bL}\bx)\cdot\nabla_{\bw} \gf
    - (\bL\bv)\cdot\nabla_{\bw} \gf
    =
    \frac{\partial \gf}{\partial t}
    - \bigl[(\dot{\bL} + \bL^2)\bx + \bL\bw\bigr]\cdot\nabla_{\bw} \gf
    = 
    \coll[\gf,\gf].
    \label{eq:bte_substituted}
\end{equation}
where the reduced form of the Collision operator $\coll$ is:
\begin{equation}
    \coll[\gf,\gf] = \int_{\mathbb{R}^3} \int_{\mathcal{S}} \left( g'_{*} g' - g_* g\right) |\bfw-\bfw_*| \dm\mathcal{S} \dm\bfw_*
    \label{eq:red_coll}
\end{equation}

We next notice that for flows $\bL(t)=\bA(\bI+t\bA)^{-1}$, we have $\dot{\bL}(t) = -\bL(t)^2$. By differentiating $\bL=\bA(\bI+t\bA)^{-1}$ with respect to $t$, we get:
\begin{equation}
  \dot{\bL}
  = \bA\,\frac{\mathrm{d}}{\mathrm{d}t}\bigl[(\bI+t\bA)^{-1}\bigr]
  = \bA\bigl[-(\bI+t\bA)^{-1}\bA(\bI+t\bA)^{-1}\bigr]
  = -[\bA(\bI+t\bA)^{-1}][\bA(\bI+t\bA)^{-1}]
  = -\bL^2,
\end{equation}

Therefore, \eqref{eq:bte_substituted} reduces to the spatially homogeneous form:
\begin{equation}
    \frac{\partial \gf}{\partial t}
    =
    \bL\bw\cdot\nabla_{\bw} \gf
    +
    \coll[\gf,\gf].
  \label{eq:reduced_bte}
\end{equation}

\subsection{Evolution of Macroscopic Quantities}
\label{sec:energy_discussion}

The number density, mean reduced velocity, and reduced second moment
are
\begin{align}
  n(t)        &= \int_{\mathbb{R}^3} g(\bw,t)\,\mathrm{d}\bw, \\
  \bV_0(t)    &= \frac{1}{n(t)}\int_{\mathbb{R}^3} \bw\,g(\bw,t)\,\mathrm{d}\bw, \\
  \bsigma(t)  &= \frac{1}{n(t)}\int_{\mathbb{R}^3} (\bw-\bV_0)\otimes(\bw-\bV_0)\,
                   g(\bw,t)\,\mathrm{d}\bw.
\end{align}

The specific internal energy is then $e(t) = \tfrac{1}{2}\operatorname{tr}\bsigma(t)$.
Further, from \eqref{eq:ansatz}, the continuum velocity $\bfV(\bfx,t)$ is related to $\bfV(\bfzero,t)$ by \cite{dayal2010}:
\begin{equation}
  \bV(\bx,t)
  = \frac{\int_{\mathbb{R}^3}\bv\,\ff(\bx,\bv,t)\,\mathrm{d}\bv}
         {\int_{\mathbb{R}^3}\ff(\bx,\bv,t)\,\mathrm{d}\bv}
  =  \frac{\int_{\mathbb{R}^3} \bfw g(\bfw,t) \dm\bfw}{\int_{\mathbb{R}^3} g(\bfw,t) \dm\bfw} + \bfA(\bfI+\bfA t)^{-1}\bfx
  = \bV(\bfzero,t) + \bL(t)\bx
  = \bV_0(t) + \bL(t)\bx
  \label{eq:mean_velocity}
\end{equation}

\paragraph{Number-density evolution.}
Integrating~\eqref{eq:reduced_bte} over $\bw$ and using
integration by parts on the transport term yields
\begin{equation}
  \frac{\mathrm{d}n}{\mathrm{d}t} = -\operatorname{tr}\bL(t)\,n(t).
  \label{eq:dndt}
\end{equation}
where we use mass conservation by the collision operator: $\int_{\mathbb{R}^3}\coll[g,g]\,\mathrm{d}\bw=0$. Equation~\eqref{eq:dndt} is the statement of local mass
conservation under the homoenergetic reduction: with $n$ spatially
uniform, continuity reduces to $\dot n=-\operatorname{tr}\bL\,n$, which
integrates to $n(t)=n_0/\det\bF(t)$.

\paragraph{Energy-moment evolution.}
Multiplying~\eqref{eq:reduced_bte} by $\tfrac{1}{2}|\bw|^2$ and integrating, using
integration by parts on the transport term and energy conservation by
elastic collisions: $\int_{\mathbb{R}^3}|\bw|^2\coll[g,g]\,\mathrm{d}\bw = 0$, gives
\begin{equation}
  \frac{\mathrm{d}}{\mathrm{d}t}\!\int_{\mathbb{R}^3}\!\tfrac{1}{2}|\bw|^2 g\,\mathrm{d}\bw
  = -n\,\bD(t)\!:\!\bsigma(t)
    - \operatorname{tr}\bL(t)\!\int_{\mathbb{R}^3}\!\tfrac{1}{2}|\bw|^2 g\,\mathrm{d}\bw,
  \label{eq:dEdt}
\end{equation}
where $\bD(t) = \tfrac{1}{2}(\bL+\bL^\T)$ is the strain-rate tensor.

\paragraph{Specific energy.}
Writing the kinetic energy per unit volume as $E(t) = \int_{\mathbb{R}^3}\tfrac{1}{2}|\bw|^2 g\,\mathrm{d}\bw = n(t) e(t)$, which defines
the specific internal energy $e(t)=\tfrac12\operatorname{tr}\bsigma(t)$
per particle, and differentiating, then combining with \eqref{eq:dndt}
and \eqref{eq:dEdt} gives:
\begin{equation}
  \frac{\mathrm{d}e}{\mathrm{d}t} = -\bD(t):\bsigma(t).
  \label{eq:dedt}
\end{equation}
The molecular mass enters only through the kinetic energy density $m\,n(t)\,e(t)$; it does not appear in the per-particle energy $e(t)$.
The compressibility of the flow enters only through $\deriv{n}{t}$.
For the flows with $\operatorname{tr}\bA = 0$ analyzed in \Cref{sec:results}, both $\bD(0) = \tfrac{1}{2}(\bA+\bA^\T)$ and
$\bsigma(0) = T_0\bI$ are such that
$\bD(0):\bsigma(0) = T_0\operatorname{tr}\bD(0) = 0$, so $e(t)$ has zero slope at $t=0$.
As $\bsigma$ evolves, $-\bD:\bsigma$ becomes strictly positive and $e(t)$ grows monotonically, consistent with the numerical observations.

\subsection{Affine Flows}
\label{sec:flows}

We study four affine flows, defined by different choices of $\bA$.
The symmetric (rate-of-strain) and antisymmetric (spin) parts of $\bA$ are $\bD_0:=\bfD(t=0)=\tfrac{1}{2}(\bA+\bA^\T)$ and $\bW_0:=\bfW(t=0)=\tfrac{1}{2}(\bA-\bA^\T)$\footnote{
Note that at $t = 0$, we have $\bfF=\bfI$ so the reference and spatial configurations coincide and $\bL=\bA$, making $\bfA$ a spatial velocity gradient at that instant. 
}; we use them to characterize the flows in
Table~\ref{tab:flows}.
These flows are visualized in \Cref{fig:flow_schematics} through the deformation of a reference cube of material.
For simple shear, bi-directional shear, and vortex flows, we have $\det\bF = 1$ throughout; for pressure shear,
$\det\bF(t) = 1 - 0.25t < 1$, reflecting the compression along
$x_1$.
\Cref{fig:flow_schematic_vortex_velocity} shows the instantaneous velocity field $\bV=\bL(t)\,\bx$ for the vortex-like flow at two times; the view is aligned with the vorticity axis to reveal the rotational character of the flow.

\begin{table}[htb!]
\centering
\caption{Test flows: nonzero entries of $\bA$ (all unspecified entries are zero), nonzero entries of spin $\bfW$.}
\label{tab:flows}
\scalebox{0.85}{%
\renewcommand{\arraystretch}{1.4}
\begin{tabular}{l |l| l}
\toprule
Flow
  & Nonzero entries of $\bA$
  & Nonzero $W_{ij}$
\\
\midrule
Simple shear
  & $A_{12}=+0.8$
  & $W_{12}=+0.40$
\\[4pt]
Pressure shear
  & $A_{11}=-0.25$, $A_{13}=+1.4$
  & $W_{13}=+0.70$
\\[4pt]
Bi-dir.\ shear
  & $A_{13}=+1.4$, $A_{21}=+0.9$, $A_{23}=+0.7$
  & $W_{12}=-0.45$, $W_{13}=+0.70$, $W_{23}=+0.35$
\\[4pt]
Vortex
  & $A_{13}=-1.3$, $A_{21}=+1.3$, $A_{23}=+0.7$
  & $W_{12}=-0.65$, $W_{13}=-0.65$, $W_{23}=+0.35$
\\
\bottomrule
\end{tabular}%
}
\end{table}

\begin{figure}[htbp]
	\subfloat[Deformation under simple shear ($A_{12}=0.8$). The flow is volume-preserving and the cube is sheared in the $x_1$ direction by an amount proportional to $x_2$, with no deformation in the $x_3$ direction.]{\label{fig:flow_schematics_simple_shear}\includegraphics[width=0.47\textwidth]{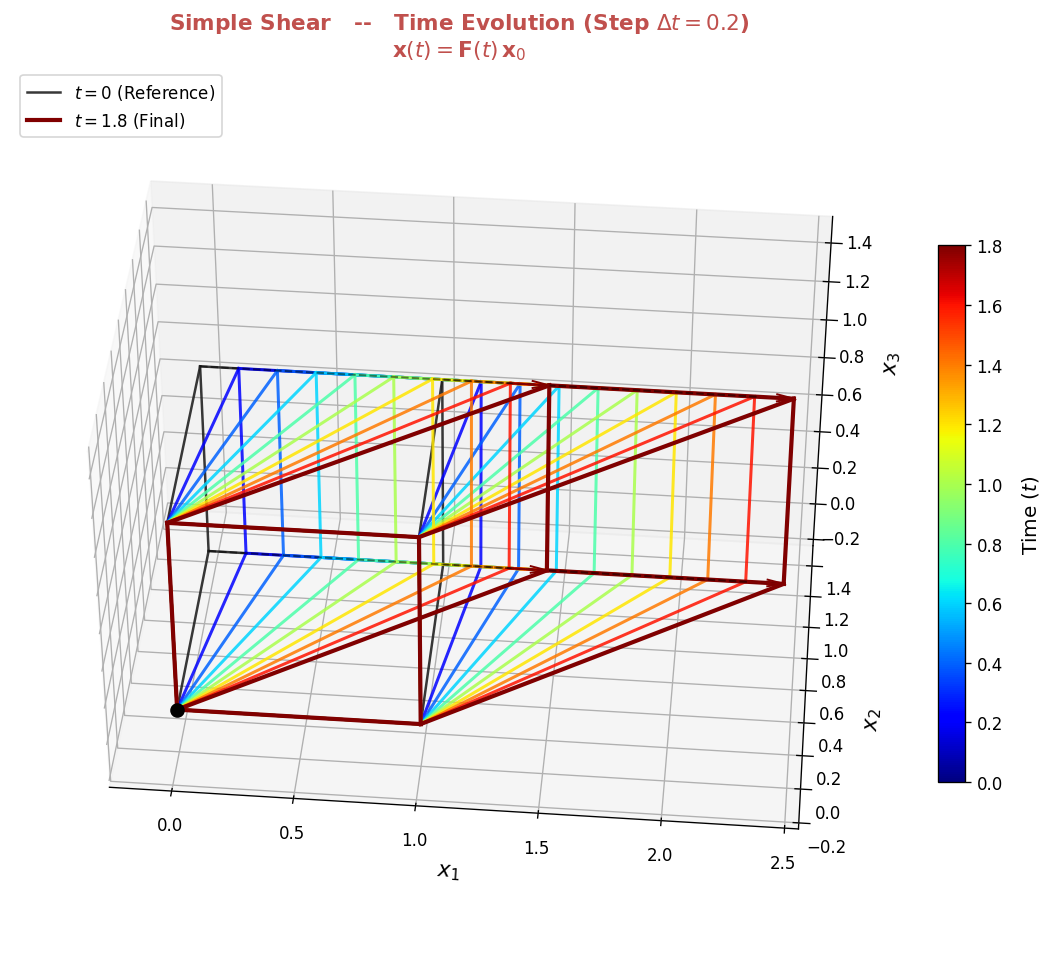}}
	\hfill
	\subfloat[Deformation under pressure shear ($A_{11}=-0.25$, $A_{13}=1.4$). The simultaneous compression along $x_1$ (visible as the shrinking of the cube in that direction making $\det(\bF)=0.55$ at $t=1.8$) and shear in the $x_1$-$x_3$ plane (visible as the progressive tilt of the faces) are discussed in \Cref{sec:dilatative}.]{\label{fig:flow_schematics_pressure_shear}\includegraphics[width=0.47\textwidth]{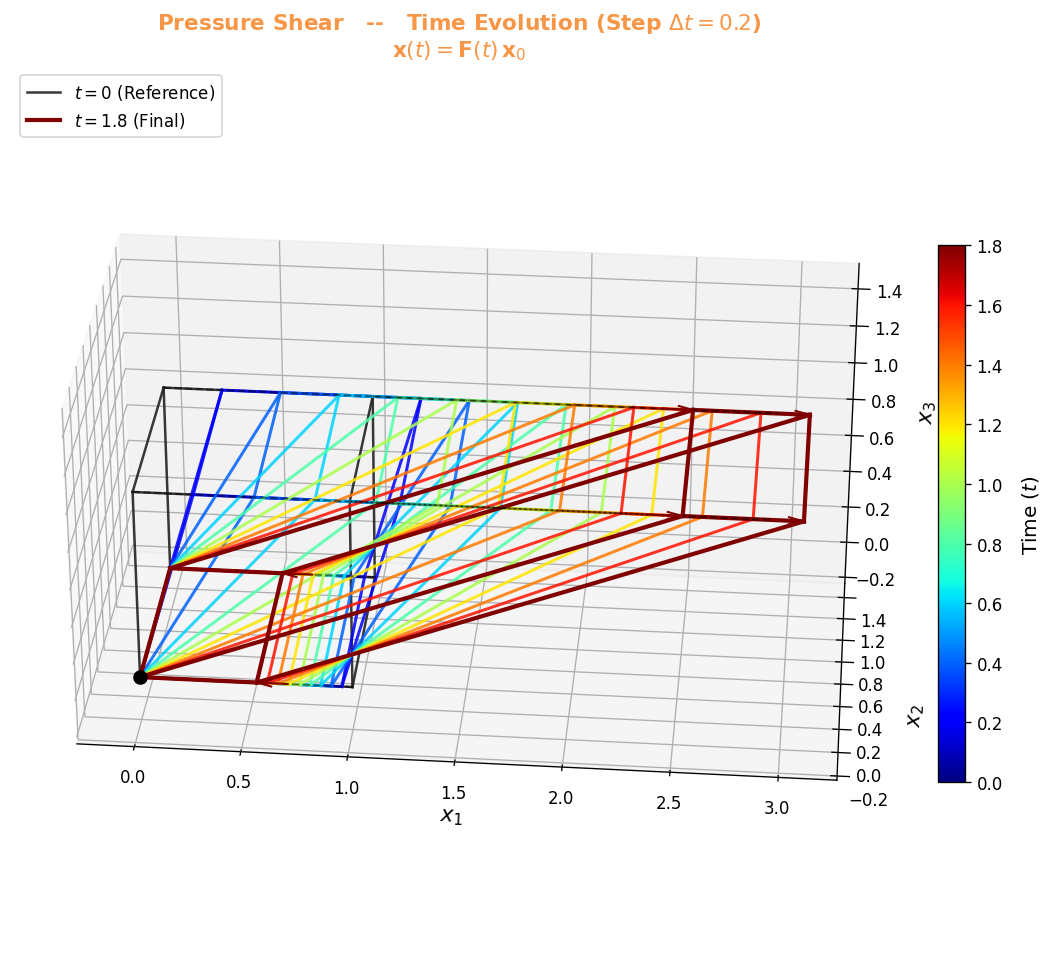}}
	\\
	\subfloat[Deformation under bi-directional shear ($A_{13}=1.4$, $A_{21}=0.9$, $A_{23}=0.7$). The three simultaneously active shear couplings produce a three-dimensional deformation: the cube develops non-trivial tilt in all three coordinate planes, in contrast to the planar character of simple and pressure shear.]{\label{fig:flow_schematics_bidir}\includegraphics[width=0.47\textwidth]{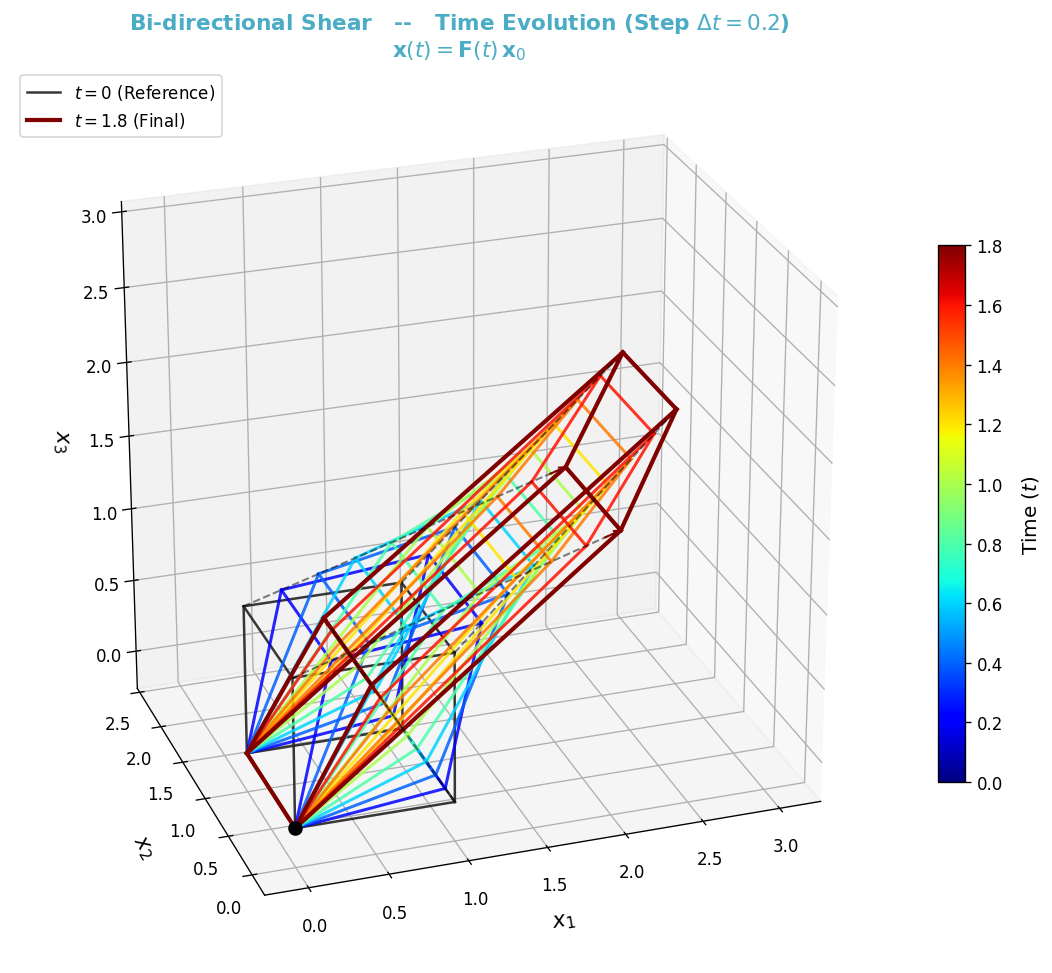}}
	\hfill
	\subfloat[Deformation under the vortex-like flow ($A_{13}=-1.3$, $A_{21}=1.3$, $A_{23}=0.7$). The dominant antisymmetric coupling ($|W_{12}|=|W_{13}|=0.65$) produces a pronounced rotational character: successive snapshots
    trace a near-circular arc in the $x_1$-$x_2$ plane, while the residual shear $A_{23}=0.7$ introduces a secondary tilt out of that plane.]{\label{fig:flow_schematics_vortex_deformation}\includegraphics[width=0.47\textwidth]{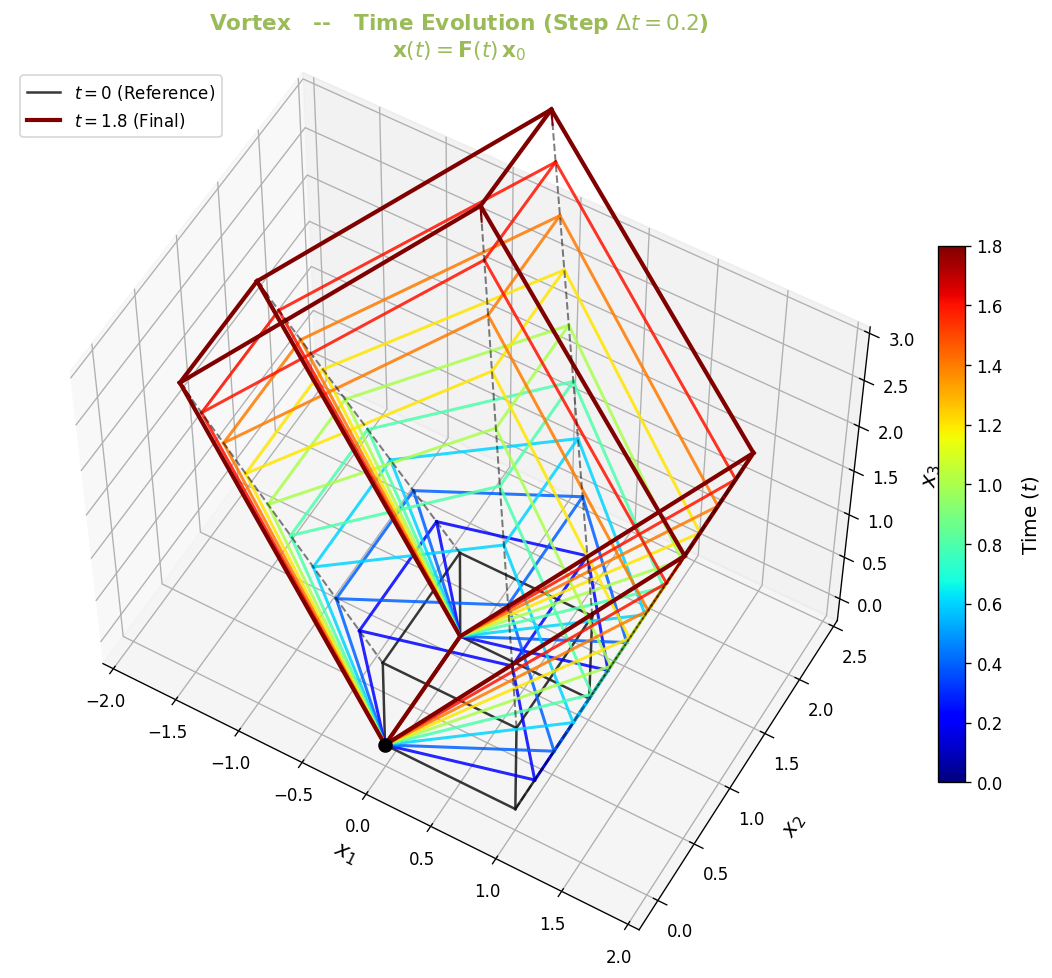}}
	\caption{Deformation of a reference cube at $t=0$ (dark grey wireframe) under the different flows at intervals of $\Delta t = 0.2$.}
	\label{fig:flow_schematics}
\end{figure}

\begin{figure}[htbp]
  \centering
  \includegraphics[width=\textwidth]{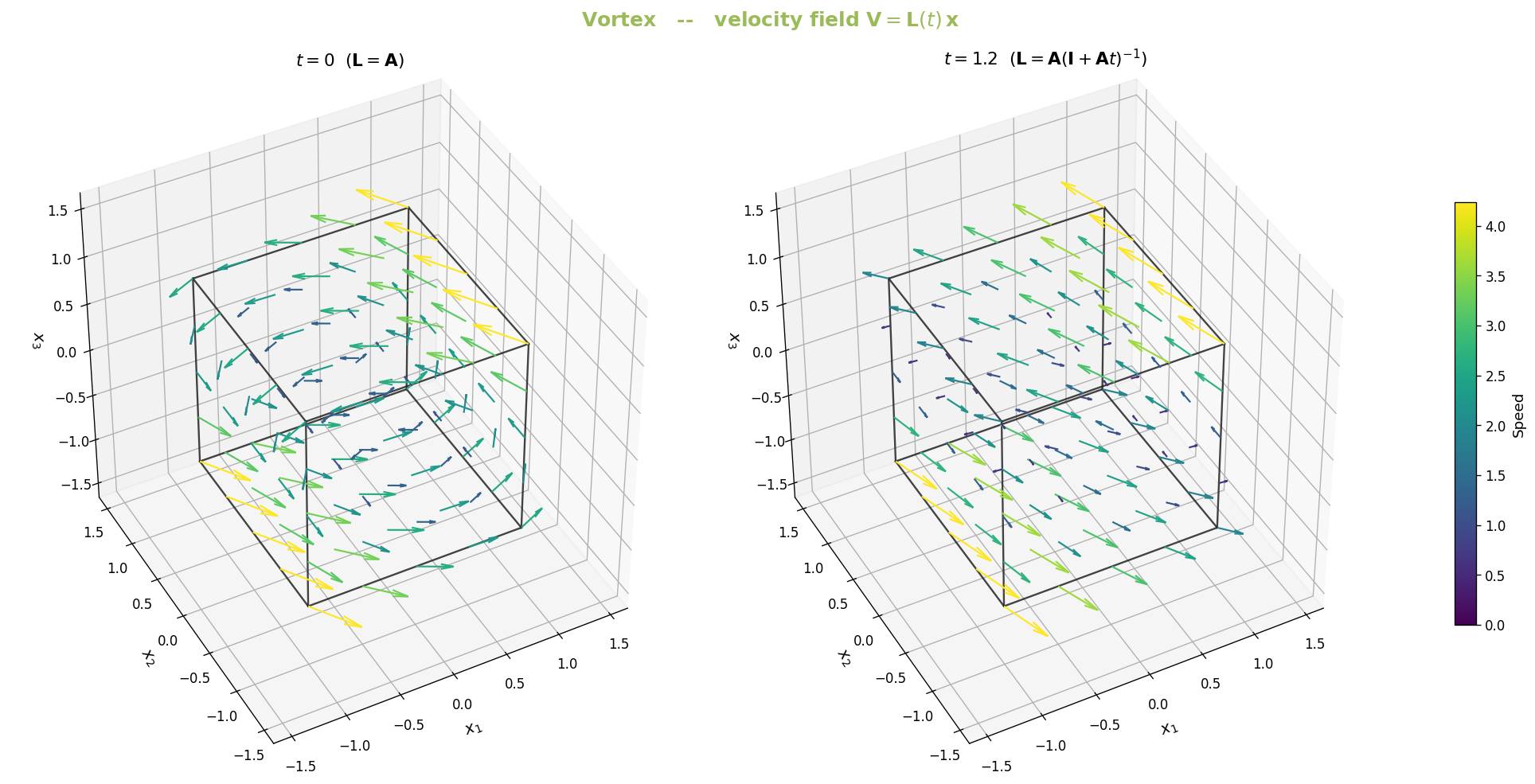}\hfill
  \caption{Instantaneous velocity field $\bV = \bL(t)\,\bx$ for the
    vortex-like flow at $t=0$ (left) and $t=1.2$ (right), evaluated
    on a uniform $2\times2\times2$ grid of sample points inside a
    cubic box (black wireframe).
    Arrows are normalized to uniform length and colored by speed.    
    The view is directed along the vorticity axis
    $\hat{\boldsymbol{\omega}} \approx (-0.36, -0.66, 0.66)$ so that
    the rotational circulation appears in the perpendicular plane.
    The field at $t=1.2$ differs visibly from $t=0$, illustrating
    the time-dependence of $\bL(t) = \bA(\bI+\bA t)^{-1}$ even when
    $\bA$ is constant.}
  \label{fig:flow_schematic_vortex_velocity}
\end{figure}

Lengths are non-dimensionalized by the mean free path $\lambda$, velocities by the most probable thermal speed
$v_\mathrm{th}=\sqrt{2k_BT_0/m}$, and time by the mean collision time $\tau=\lambda/v_\mathrm{th}$.
With this choice, the non-dimensional Boltzmann equation retains the form~\eqref{eq:reduced_bte} with the collision frequency $\nu\sim 1$.
The entries of $\bA$ in \Cref{tab:flows} are therefore the non-dimensional velocity gradients.
\section{Numerical Method}
\label{sec:numerics}
We use the Discontinuous Galerkin Finite Element Method (DG-FEM) to solve the invariance-reduced Boltzmann equation \eqref{eq:reduced_bte}, following recent work that has shown the efficacy of DG-FEM for the Boltzmann equation as well as moment closure equations based on the Boltzmann equation \cite{alekseenko2014,gamba2009,gamba2010,alekseenko2011,alekseenko2012,alekseenko2014,Barth2006Boltzmann}.
Other alternatives are Fourier-based methods \cite{pareschi2000numerical,mouhot2006fast,gamba2009} and Monte Carlo methods \cite{schwartzentruber2016nonequilibrium,bird1994,wagner1992monte,babovsky1989convergence}.

The key computational expense is in computing the collision operator.
It has been shown by \cite{alekseenko2014} that the use of symmetries and other techniques can reduce the computation from $O(n^8)$ for the brute force computation of the Boltzmann collision operator to $O(n^5)$, where $n$ is the number of quadrature points in each direction in velocity space.
Many aspects here follow closely their method, and refer to their work for details. 
We describe here only the overall structure and key points of difference. The key computational expense is in computing the collision operator.
It has been shown by \cite{alekseenko2014} that the use of symmetries
and other techniques can reduce the computation from $O(n^8)$ for the
brute force computation of the Boltzmann collision operator to
$O(n^5)$, where $n$ is the number of quadrature points in each
direction in velocity space.
Many aspects here follow closely their method, and we refer to their
work for details.
We describe here only the overall structure and key points of
difference. The principal points of difference are the following.
First, whereas \cite{alekseenko2014} solve the spatially homogeneous
Boltzmann equation, we solve the homoenergetic reduction, in which the
affine deformation enters through the transport term
$-(\bL(t)\bw)\cdot\nabla_\bw\gf$; this term is discretized in the same
reduced velocity variable $\bw$ as the collision operator, so that no
separate physical-space advection step is required and the strongly
anisotropic distributions produced by the flow are resolved on the
same mesh.
Second, the constrained conservation
correction (Sec.~\ref{sec:conservation}) is applied selectively: for the
driven affine flows only mass is enforced, with momentum and energy left
unconstrained, since the deformation continuously supplies both. Third,
for compressible flows the mass constraint is made time-dependent
(Sec.~\ref{sec:conservation_fix}), targeting $n_0/\det\bF(t)$ so that the
density follows the compression or expansion of the flow. These choices
adapt the methods of
\cite{alekseenko2014,zhang2018} to the time-dependent deformations studied here.

\subsection{DG-FEM velocity-space discretisation}
\label{sec:dg}

\paragraph*{Computational domain.}
The reduced velocity domain is truncated to $\Omega =[-w_\text{max}, w_\text{max}]^3$ with $w_\text{max}=3$ (in thermal-speed units).
We use a Maxwellian initial condition and the simulation is halted before the distribution function reaches the domain boundary to avoid domain-truncation artefacts.

\paragraph*{Mesh.}
\label{sec:mesh}
$\Omega$ is partitioned into $N_e^3$ cubic elements of equal side
length $h = 2w_\text{max}/N_e$.
Two meshes are used:
\begin{itemize}
  \item $N_e = 3$: $3^3 = 27$ elements, $h = 2$.
  \item $N_e = 5$: $5^3 = 125$ elements, $h = 1.2$.
\end{itemize}
Within each element $\Omega^k$, a local coordinate
$\bxi = (\xi_1,\xi_2,\xi_3)\in[-1,1]^3$ is introduced via the
affine map
\begin{equation}
  w_i = w_i^{(k)} + \frac{h}{2}\,\xi_i,
  \label{eq:affine_map}
\end{equation}
where $w_i^{(k)}$ is the $i$-th coordinate of the centre of $\Omega^k$.

\paragraph*{Basis functions.}
\label{sec:basis}

On each element $\Omega^k$ we use the tensor-product Lagrange basis of
degree $p=2$ (quadratic, three Gauss--Legendre nodes per direction):
\begin{equation}
  \phi_{\bfalpha}(\bxi)
  = \prod_{k=1}^{3} \ell_{\alpha_k}(\xi_k),
  \qquad \alpha_k \in \{1,2,3\},
  \label{eq:basis}
\end{equation}
where $\ell_j$ is the $j$-th Lagrange polynomial on the three
Gauss--Legendre nodes $\{-\sqrt{3/5},\,0,\,+\sqrt{3/5}\}$.
Each element has $3^3 = 27$ degrees of freedom.
The global approximation space $V_h$ has dimension $27N_e^3$
(729 for the $3^3$ mesh, 3375 for the $5^3$ mesh).

The approximate solution restricted to element $\Omega^k$ is
\begin{equation}
  \gf_h(\bw,t)\big|_{\Omega^k}
  = \sum_{|\bfalpha|=1}^{27}
    g_{\bfalpha}^{(k)}(t)\,\phi_{\bfalpha}\!\left(\Xi^k(\bw)\right),
  \label{eq:DGapprox}
\end{equation}
where $\Xi^k(\bw)$ is the affine map~\eqref{eq:affine_map} from
physical to reference coordinates and $g_{\bfalpha}^{(k)}(t)$ are
the time-dependent nodal coefficients.

\subsection{Weak formulation}
\label{sec:weak}

Multiplying~\eqref{eq:reduced_bte} by a test function $\phi\in V_h$,
integrating over an element $\Omega^k$, and integrating the drift term
by parts gives:
\begin{equation}
  \int_{\Omega^k} \phi\,\frac{\partial\gf_h}{\partial t}\,\mathrm{d}\bw
  = -\int_{\Omega^k} L_{ij}\,w_j\,\gf_h\,\frac{\partial\phi}{\partial w_i}
      \,\mathrm{d}\bw
  + \int_{\partial \Omega^k} \widehat{\mathrm{\Phi}}_n\,\phi
      \,\mathrm{d}S
  + \int_{\Omega^k} \phi\,\coll[\gf_h,\gf_h]\,\mathrm{d}\bw,
  \label{eq:weak}
\end{equation}
where $\widehat{\mathrm{\Phi}}_n$ denotes the numerical flux.

Using the DG-FEM approximation $\gf_h = \sum_j g_j(t)\phi_j(\bw)$ and $\phi_i$ as test functions, we can follow the usual finite element manipulations to arrive at the final evolution equations in the discretized form:
\begin{equation}
\label{eqn:discrete-evolution}
	\sum_j P_{ij} \dot{g}_j(t) =  - \Psi_i + \trace\left(\bfL\right) \sum_j P_{ij} g_j(t)  + \sum_j Y_{ij}g_j(t) + \sum_{k,l} Z_{ilk}g_l(t) g_k(t) 
\end{equation}
where the operators are:
\begin{equation}
  P_{ij} = \int_{\Omega^k} \phi_j\phi_i\,\mathrm{d}\bw,
  \qquad
  \Psi_i = \int_{\partial\Omega^k}\phi_i\widehat{\mathrm{\Phi}}_n\,\mathrm{d}\bw,
  \qquad
  Y_{ij}(t) = \int_{\Omega^k} \phi_j\,
               \frac{\partial\phi_i}{\partial w_r}\,L_{rm}(t)\,w_m
               \,\mathrm{d}\bw,
  \label{eq:Bw_ops}
\end{equation}
\begin{equation}
\begin{aligned}
  Z_{ijk}
  &=
  \int_{\Omega^k}\!\int_{\Omega^k}
            \frac{|\bw-\bw_*|}{2}
            \Bigl(\int_{\mathcal{S}}
              (\phi'_{*i}+\phi'_i-\phi_{*i}-\phi_i)\,\mathrm{d}S
            \Bigr)
            \phi_j(\bw_*)\,\phi_k(\bw)\,\mathrm{d}\bw_*\,\mathrm{d}\bw \\
  &=
  \int_{\Omega^k}\!\int_{\Omega^k}
            M_i\phi_j(\bw_*)\,\phi_k(\bw)\,\mathrm{d}\bw_*\,\mathrm{d}\bw ,
\end{aligned}
\label{eq:Z_op}
\end{equation}
and $M_i$ is:
\begin{equation}
    M_i =\frac{|\bw-\bw_*|}{2}
            \int_{\mathcal{S}}
              (\phi'_{*i}+\phi'_i-\phi_{*i}-\phi_i)\,\mathrm{d}S
\end{equation}

The operators have the following properties:
\begin{enumerate}
\item \textit{Sparsity and locality.} The DG basis restricts all integrations to a single element; $P, Y$, and $Z$ are accordingly sparse.
\item \textit{Uniformity.} On a uniform mesh, the operators are identical across elements, so they need only be computed once.
\item \textit{Diagonal mass matrix.} Orthogonality of the Lagrange basis makes $P$ diagonal, enabling direct inversion in the time-stepping scheme~\eqref{eqn:discrete-evolution}.
\item \textit{Time independence of $P$ and $Z$.} Both are independent of
$t$ and are precomputed and stored once; $Y$ depends on $t$ through $\bL$(t) and is recomputed each step.
\item \textit{Collision operator compression.} The symmetry properties of $M$, as discussed in \cite{alekseenko2014} lets us reduce the memory storage for the collision operator dramatically. In case of a uniform cubical mesh, all we need to do is to compute the collision operator for a single element (in our case, we compute for the central element of the domain) and the values for the rest of the elements could be obtained using the invariance property.
\item \textit{Quadrature accuracy.} In order to achieve an accuracy of $10^{-8}$ in the collision operator approximation, which is required for accurate time integration that predicts the correct relaxation times and preserves the lower moments, we use adaptive quadrature for calculating $M_i$ following \cite{alekseenko2014}.
\end{enumerate}

The term $\Psi_i$ is an integration over the element boundaries to introduce the numerical flux. 
Unlike classical FEM where inter-element continuity is enforced through the choice of shape functions, DG-FEM enforces inter-element continuity in a weak sense; in hyperbolic problems, this corresponds to a numerical flux. 
Following \cite{Barth2006Boltzmann,lax1954}, we use the upwind form of the Lax--Friedrichs numerical flux, denoted by $\widehat{\mathrm{\Phi}}_n$, for the inter-element communication in the boundary term $\Psi_i$:
\begin{equation}
  \widehat{\mathrm{\Phi}}_n = \frac{p_n}{2}\left(g^-+g^+\right)
+
\frac{|p_n|}{2}\left(g^- - g^+\right),
\qquad
p_n = L_{rs}(t)w_s \hat{n}_r,
  \label{eq:flux}
\end{equation}
where $g^\pm$ denote the traces from the interior ($-$) and exterior
($+$) sides of the face, and $\hat{n}_r$ is the outward unit normal
from the current element.
This form of the Lax--Friedrichs upwind flux is
obtained by setting $\alpha=0$ in the general formula of
\cite{hesthaven2007nodal}. Equivalently, \eqref{eq:flux} can be written as:
\begin{equation}
\widehat{\mathrm{\Phi}}_n
=
\begin{cases}
p_n g^-, & p_n\ge 0,\\
p_n g^+, & p_n<0.
\end{cases}
\label{eq:flux_simple}
\end{equation}

\subsection{Discrete evaluation of the collision tensor.}
\label{sec:collision}
The collision tensor $Z_{ijk}$ is the dominant computational cost
of the method, but it is a one-time precomputation: because
$Z_{ijk}$ is independent of both $t$ and $\gf$, it is assembled once
before the time loop and stored.
Each time step then requires only the evaluation of the matrix--vector
product:
\begin{equation}
  \dot{g}_i(t)
  = P_{ij}^{-1}
    \Bigl[
      Y_{jk}(t)\,g_k(t)
      + \operatorname{tr}(\bL(t))\,P_{jk}\,g_k(t)
      - \Psi_j
      + Z_{jkl}\,g_k(t)\,g_l(t)
    \Bigr],
  \label{eq:gdot}
\end{equation}

The key reduction follows \cite{alekseenko2014}: the use of symmetry
properties of the collision operator reduces the nominal $O(n^8)$
cost of the brute-force integral to $O(n^5)$, where $n$ is the number
of quadrature points per direction in velocity space.
This reduction exploits the orthogonality of the Lagrange basis
(Lemma~3.1 in \cite{alekseenko2014}), which means the volume integrals
over $\bw$ and $\bw_*$ in $Z_{ijk}$ reduce to products of Gauss
quadrature weights evaluated at the nodal positions, leaving only the
angular integral over $\mathcal{S}^2$ to be performed numerically.
Adaptive quadrature over $\mathcal{S}^2$ is used following
\cite{alekseenko2014} to achieve a collision operator accuracy of
$10^{-8}$, which is required to preserve lower moments and predict
correct relaxation times.

The translation-invariance of the hard-sphere kernel
(Lemma~4.2 in \cite{alekseenko2014}) means that on a uniform mesh
$Z_{ijk}$ need only be computed for a single (central) element;
the entries for all remaining elements are recovered by the
invariance property, substantially reducing both memory and
precomputation cost.
The local support of the DG basis functions additionally makes
$Z_{ijk}$ sparse, since most pairs $(\bw_j, \bw_k)$ have
non-overlapping collision spheres with respect to most test
functions $\phi_i$.

\subsection{Conservation correction}
\label{sec:conservation}

The discrete collision operator $Z_{ijk}$ does not exactly conserve
moments due to truncation of the velocity domain at $|\bw|=3$.
Following Zhang and Gamba \cite{zhang2018}, a constrained least-squares
correction is applied after each collision evaluation.
The corrected operator $Z_c$ is the solution of
\begin{equation}
  \min_{Z_c}\;\tfrac{1}{2}(Z_c-Z)^\T B(Z_c-Z),
  \qquad \text{subject to}\quad QF_c = 0,
  \label{eq:conservation_opt}
\end{equation}
where $Q$ is a constraint matrix whose rows select the moments to be
enforced.
The solution is
\begin{equation}
  Z_c = \Bigl[\bI - B^{-1}Q^\T(QB^{-1}Q^\T)^{-1}Q\Bigr]Z.
  \label{eq:conservation_soln}
\end{equation}
An important distinction between equilibrium and non-equilibrium settings applies here.
For the spatially homogeneous relaxation problem ($\bA=\bfzero$), $Q$ contains constraints for mass, momentum, and all three second moments (energy).
For the affine flows studied here ($\bA\neq\bfzero$), only mass conservation is enforced; momentum and energy are left unconstrained, because the applied deformation continuously transfers momentum and energy to the system.

\subsection{Time-dependent conservation constraint for compressible flows}
\label{sec:conservation_fix}

For an incompressible flow with $\operatorname{tr}\bA=0$, we have $\mathrm{d}n/\mathrm{d}t = -\operatorname{tr}\bL(t)\,n \implies n(t) = n_0$ for all $t$.
The projection described in \Cref{sec:conservation}
enforces  $\int g_h\,\mathrm{d}\bw = n_0$, holding the number density fixed
at its initial value and corrects against discretization errors.

For compressible flows with $\operatorname{tr}\bA \neq 0$, however,
the number density evolution is:
\begin{equation}
  n(t) = \frac{n_0}{\det\bF(t)} = \frac{n_0}{J(t)},
  \label{eq:n_evolution}
\end{equation}
where $J(t) = \det\bF(t)$ is the Jacobian of the deformation
gradient.
We therefore modify the conservation constraint to target the
time-dependent value:
\begin{equation}
  \int g_h(\bw, t)\,\mathrm{d}\bw = \frac{n_0}{J(t)}.
  \label{eq:conservation_target}
\end{equation}
Equivalently, the constraint enforces
$\mathrm{d}n/\mathrm{d}t = -\operatorname{tr}\bL(t)\,n(t)$.

For the pressure-shear flow studied in \Cref{sec:dilatative},
the modified constraint allows for compressibility and \Cref{fig:dilatative_n} shows that the resulting number density tracks the closed-form expression $n_0/(1+0.3t)$ to within numerical precision.

\subsection{Time integration}
\label{sec:time}

The system of ODEs for the nodal coefficients $\{g_{\bm\alpha}^{(e)}(t)\}$ obtained from~\eqref{eq:weak} is integrated using a two-stage explicit multi-step scheme \cite{hairer1993}.
The first four time steps use a fifth-order Runge--Kutta method to generate the startup history required by the multi-step integrator.
All subsequent steps use the fifth-order Adams--Bashforth formula:
\begin{equation}
  \bm{g}^{n+1} = \bm{g}^n
    + \frac{\Delta t}{720}\bigl[
        1901\,\bm{R}^n
      - 2774\,\bm{R}^{n-1}
      + 2616\,\bm{R}^{n-2}
      - 1274\,\bm{R}^{n-3}
      +  251\,\bm{R}^{n-4}
      \bigr],
  \label{eq:AB5}
\end{equation}
where $\bm{R}^k = \bm{R}(\bm{g}^k, t^k)$ collects the right-hand side of~\eqref{eq:weak} for all degrees of freedom at step $k$.
The time step is $\Delta t = 10^{-3}$.

\subsection{Code verification}
\label{sec:verification}

The DG-FEM implementation was verified in two stages that test the collision
operator and the transport operator independently; full details and
figures are given in \cite{debnath2018}.

\paragraph{Collision operator: spatially homogeneous relaxation.}
Setting $\bA=\bm{0}$ eliminates the transport term and isolates the
collision operator.
Two experiments were run on both meshes.
Experiment~1 uses a sum of two shifted Maxwellians at different
temperatures and bulk velocities as the initial condition:
\begin{equation}
  g(0,\bw) = \frac{\rho_1}{(2\pi T_1)^{3/2}}
              \exp\!\Bigl(-\tfrac{|\bw-\bv_1|^2}{2T_1}\Bigr)
           + \frac{\rho_2}{(2\pi T_2)^{3/2}}
              \exp\!\Bigl(-\tfrac{|\bw-\bv_2|^2}{2T_2}\Bigr).
  \label{eq:Iw_relax}
\end{equation}
Experiment~2 uses discontinuous (top-hat) initial data.
In both cases the solution converges to a single isotropic Maxwellian,
consistent with Boltzmann's $H$-theorem.
Mass, momentum, and energy are conserved throughout to relative
deviations below $10^{-6}$, confirming the correctness of the
collision operator computation and the conservation correction of
Section \ref{sec:conservation}.

\paragraph{Transport operator: method of manufactured solutions.}
The transport term of~\eqref{eq:reduced_bte} was verified independently using the method of manufactured solutions \cite{roache2002}.
The manufactured solution
\begin{equation}
  g_\text{ms}(t,\bw) = \sin(t)\exp(-|\bw|^2)
  \label{eq:mms_soln}
\end{equation}
satisfies the initial and boundary conditions for all $t$.
Substituting into $\partial_t g + \bp\cdot\nabla_{\bw} g = 0$
(where $\bp=-\bL(t)\bw$) produces the source term
\begin{equation}
  f_\text{ms}(t,\bw)
  = \exp(-|\bw|^2)\bigl(\cos t - 2\sin t\;(\bw\cdot\bp)\bigr).
  \label{eq:mms_source}
\end{equation}
Simple shear ($\bA=0.8\,\be_1\otimes\be_2$) was used for $\bL(t)$.
The fifth-order Adams--Bashforth scheme recovers~\eqref{eq:mms_soln}
with errors consistent with the expected order of accuracy,
confirming the correctness of both the DG-FEM discretisation and the
time integration.

\subsection{Mesh convergence}
\label{sec:meshconv}

Both the $3^3$ and $5^3$ meshes are run for all four flows.
The principal axis angle of the fitted Gaussian (see Section \ref{sec:kinematic}) agrees to within $1^\circ$ between the two meshes at every time step and for every flow.
The fitted eigenvalue magnitudes show a systematic offset of approximately $10$--$15\%$ between meshes; the $5^3$ mesh gives larger eigenvalues, consistent with the finer resolution capturing more anisotropy before numerical dissipation reduces it.
All quantitative results reported in Section \ref{sec:results} use the $5^3$ mesh unless otherwise noted.

\section{Free-Streaming Kinematic Prediction}
\label{sec:kinematic}

In the absence of collisions, the reduced equation~\eqref{eq:reduced_bte} reduces to the pure advection problem:
\begin{equation}
  \frac{\partial\gf}{\partial t}
  = L_{ij}(t)\,w_j\,\frac{\partial\gf}{\partial w_i},
  \label{eq:freestream_pde}
\end{equation}
The characteristics of~\eqref{eq:freestream_pde} satisfy
\begin{equation}
  \dot{\bw} = -\bL(t)\bw,
  \label{eq:char_ode}
\end{equation}
with solution
\begin{equation}
  \bw(t) = \bM(t)\,\bw(0),
  \label{eq:char_soln}
\end{equation}
where $\bM(t)$ satisfies $\dot{\bM}=-\bL\bM$, $\bM(0)=\bI$.
Along these characteristics $\gf$ is constant:
\begin{equation}
    \deriv{}{t}\gf(\bw(t),t) 
    = \parderiv{\gf}{t}
    + \dot{w}_i\parderiv{\gf}{w_i}
    = L_{ij}w_j\parderiv{\gf}{w_i} + (-L_{ij}w_j)\parderiv{\gf}{w_i} = 0    
\end{equation}

\begin{proposition}[Characteristic solution]
\label{prop:M}
For the homoenergetic velocity gradient $\bL(t)=\bA(\bI+t\bA)^{-1}$,
the characteristic map is
\begin{equation}
  \bM(t) = (\bI+t\bA)^{-1} = \bF^{-1}(t),
  \label{eq:M}
\end{equation}
where $\bF(t)=\bI+t\bA$ is the deformation gradient.
\end{proposition}
\begin{proof}
Let $\bM=(\bI+t\bA)^{-1}$.  Differentiating:
\begin{align*}
  \dot{\bM}
  &= -(\bI+t\bA)^{-1}\bA(\bI+t\bA)^{-1}.
\end{align*}
Since $\bA$ commutes with $(\bI+t\bA)$, we have $(\bI+t\bA)^{-1}\bA = \bA(\bI+t\bA)^{-1}$, so
\begin{align*}
  \dot{\bM}
  = -\bA(\bI+t\bA)^{-2}
  = -\bA(\bI+t\bA)^{-1}\cdot(\bI+t\bA)^{-1}
  = -\bL\bM.
\end{align*}
Since $\bM(0)=\bI$, this gives $\bM(t)=\bF^{-1}(t)$.
\end{proof}

Since $\gf$ is constant along characteristics and the characteristic
passing through $\bw$ at time $t$ originated at $\bw_0=\bF(t)\bw$,
the free-streaming solution is
\begin{equation}
  \gf_\text{fs}(\bw,t) = \gf_0\!\left(\bF(t)\bw\right).
  \label{eq:fs_solution}
\end{equation}

Here $\bw$ and $\bw_0$ are the values of the same characteristic in the reduced velocity space at times $t$ and $0$, both elements of the same $\mathbb{R}^3$; accordingly $\bF(t)$ enters
\eqref{eq:fs_solution} as the linear map $\bM^{-1}(t)$ on velocity
space defined in Proposition~\ref{prop:M}, not as a two-point
deformation gradient acting on a spatial vector.
The distribution function value is simply transported along characteristics unchanged.
The number density does change for compressible flows: $n(t)=\int\gf_\text{fs}\,\mathrm{d}\bw = n_0/|\det\bF(t)|$,
as verified by integrating~\eqref{eq:freestream_pde} over $\bw$.

\subsection{Evolution of the Covariance}
\label{sec:cov_pred}

For the Maxwellian initial condition $\gf_0(\bw)=(2\pi)^{-3/2}\exp(-|\bw|^2/2)$, the free-streaming solution~\eqref{eq:fs_solution} is:
\begin{equation}
  \gf_\text{fs}(\bw,t)
  = \frac{1}{(2\pi)^{3/2}}
     \exp\!\bigl(-\tfrac{1}{2}\bw^\T\bF^\T\bF\,\bw\bigr).
  \label{eq:fs_gaussian}
\end{equation}
The covariance is defined as 
$\Sigma_{mn}(t):=\frac{\int w_mw_n\gf_\text{fs}\,\dm\bw}{\int\gf_\text{fs}\,\dm\bw}$.
Substituting $\bu=\bF\bw$ (implying $\bw=\bM\bu$, $\mathrm{d}\bw=|\det\bF|^{-1}\mathrm{d}\bu$):
\begin{equation*}
  \int w_mw_n\gf_\text{fs}\,\mathrm{d}\bw
  = \frac{|\det\bF|^{-1}}{(2\pi)^{3/2}}
    M_{mi}M_{nj}\!\int u_iu_j e^{-|\bu|^2/2}\,\mathrm{d}\bu
  = |\det\bF|^{-1}M_{mi}M_{ni},
\end{equation*}
and using $\int\gf_\text{fs}\,\mathrm{d}\bw = |\det\bF|^{-1}$ gives:
\begin{equation}
  \bsigma_\text{fs}(t) = \bM\bM^\T = (\bF^\T\bF)^{-1} = [(\bI+t\bA)^\T(\bI+t\bA)]^{-1} = \bigl[\bI + 2t\bD_0 + t^2\bA^\T\bA\bigr]^{-1}.
  \label{eq:Cpred}
\end{equation}
where $\bD_0=\tfrac{1}{2}(\bA+\bA^\T)$ is the rate-of-strain tensor evaluated at $t=0$, to be distinguished from the time-dependent spatial
rate-of-strain $\bD(t)=\operatorname{sym}\bL(t)$ appearing in the
energy balance~\eqref{eq:dedt}.
At leading order in $t$, $\bsigma_\text{fs}(t)\approx\bI - 2t\bD_0$, so the initial rate of change of the covariance is controlled entirely by the strain rate $\bD_0$.

\section{Results}
\label{sec:results}

\subsection{Anisotropic Gaussian structure}
\label{sec:results_gaussian}

The central finding in this paper is that the reduced velocity distribution is close to an anisotropic Gaussian away from equilibrium.
This can be observed visually: 
\Cref{fig:g_at_1200} shows representative planar views of the isotropic gaussian velocity distribution at $t=0$ on the $w_3=0$ plane, and the development of anisotropy at $t=1.21$,
and \Cref{fig:g_evolution} shows the evolution of $g$ on cross-sectional slices.

\begin{figure}[htbp]
  \centering
  \subfloat[Initial condition\label{fig:in_cond}]{\includegraphics[width=0.50\textwidth]{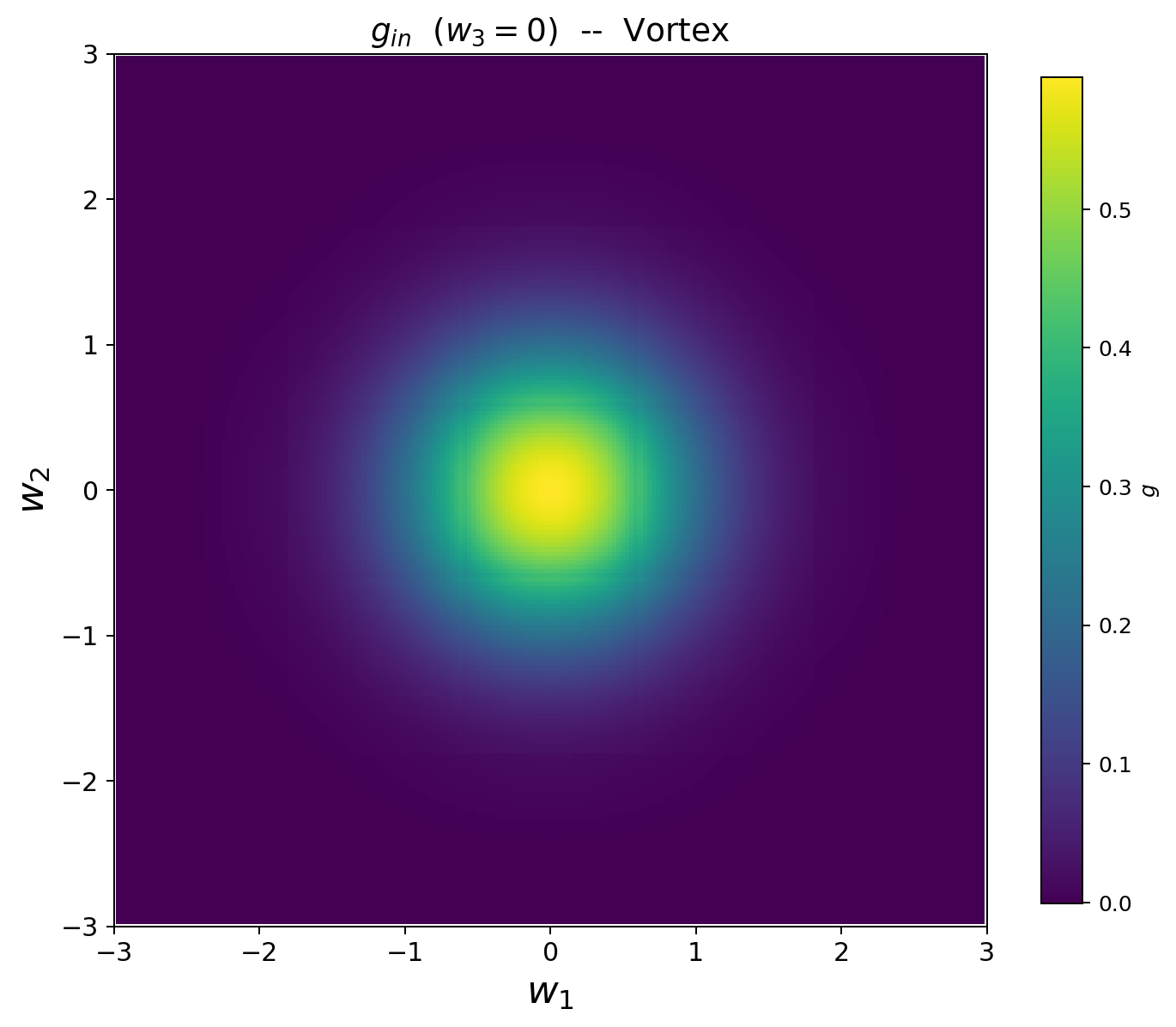}}\hfill
  \subfloat[$g_{out}$ at $t=1.21$ on plane $w_3=0$\label{fig:g_out_w12}]{\includegraphics[width=0.50\textwidth]{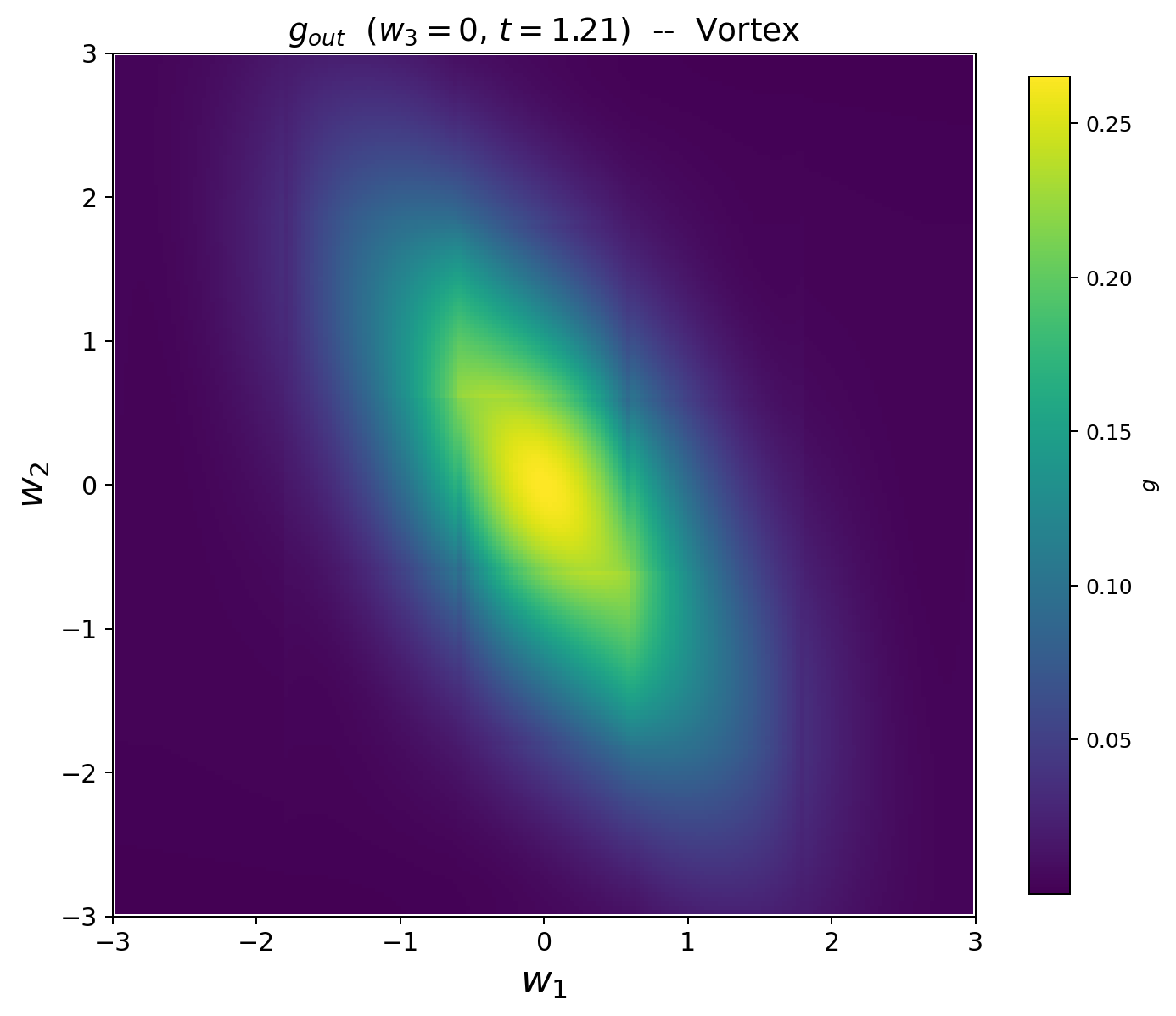}}\\[1ex]
  \subfloat[$g_{out}$ at $t=1.21$ on plane $w_2=0$\label{fig:g_out_w13}]{\includegraphics[width=0.50\textwidth]{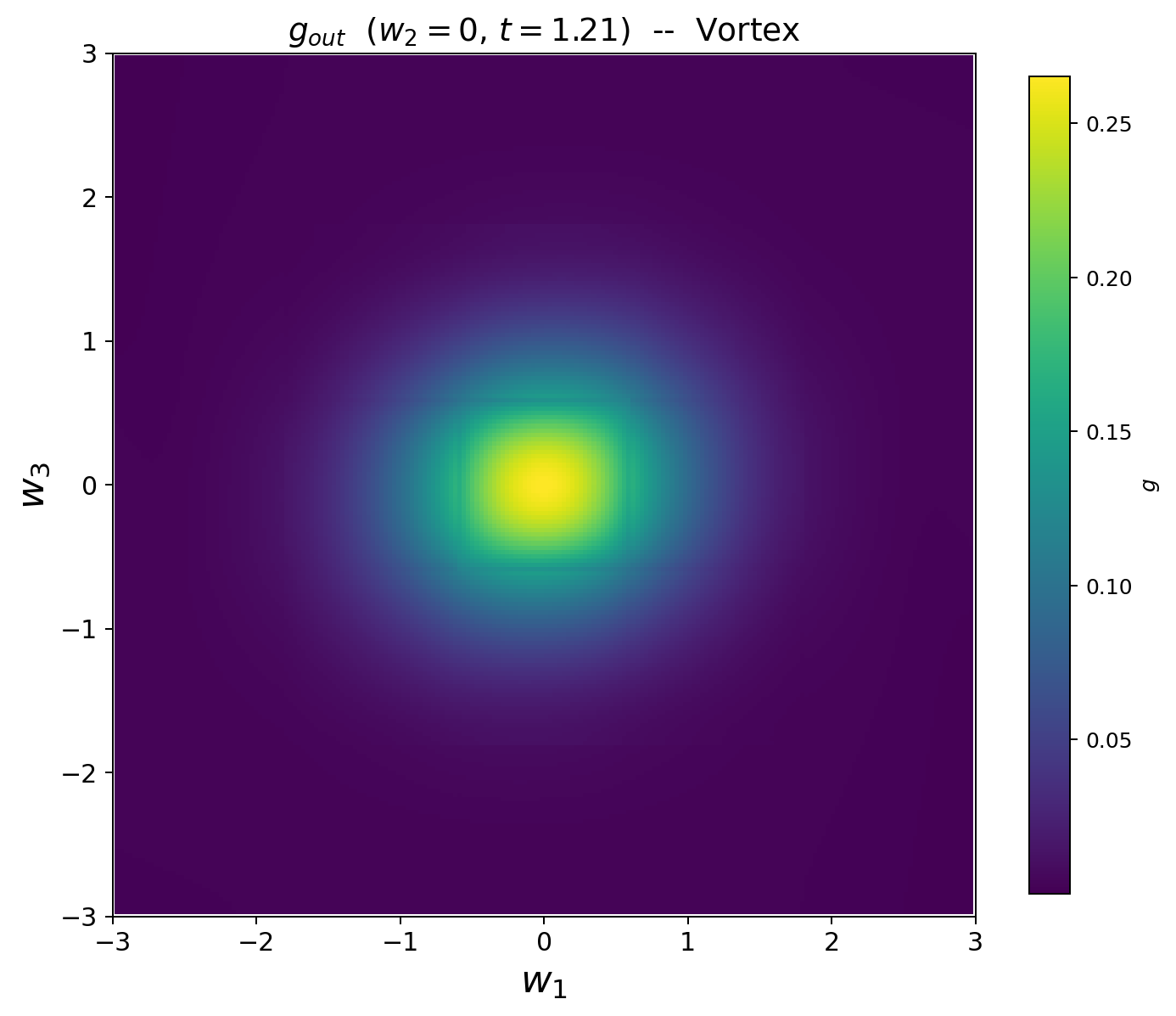}}\hfill
  \subfloat[$g_{out}$ at $t=1.21$ on plane $w_1=0$\label{fig:g_out_w23}]{\includegraphics[width=0.50\textwidth]{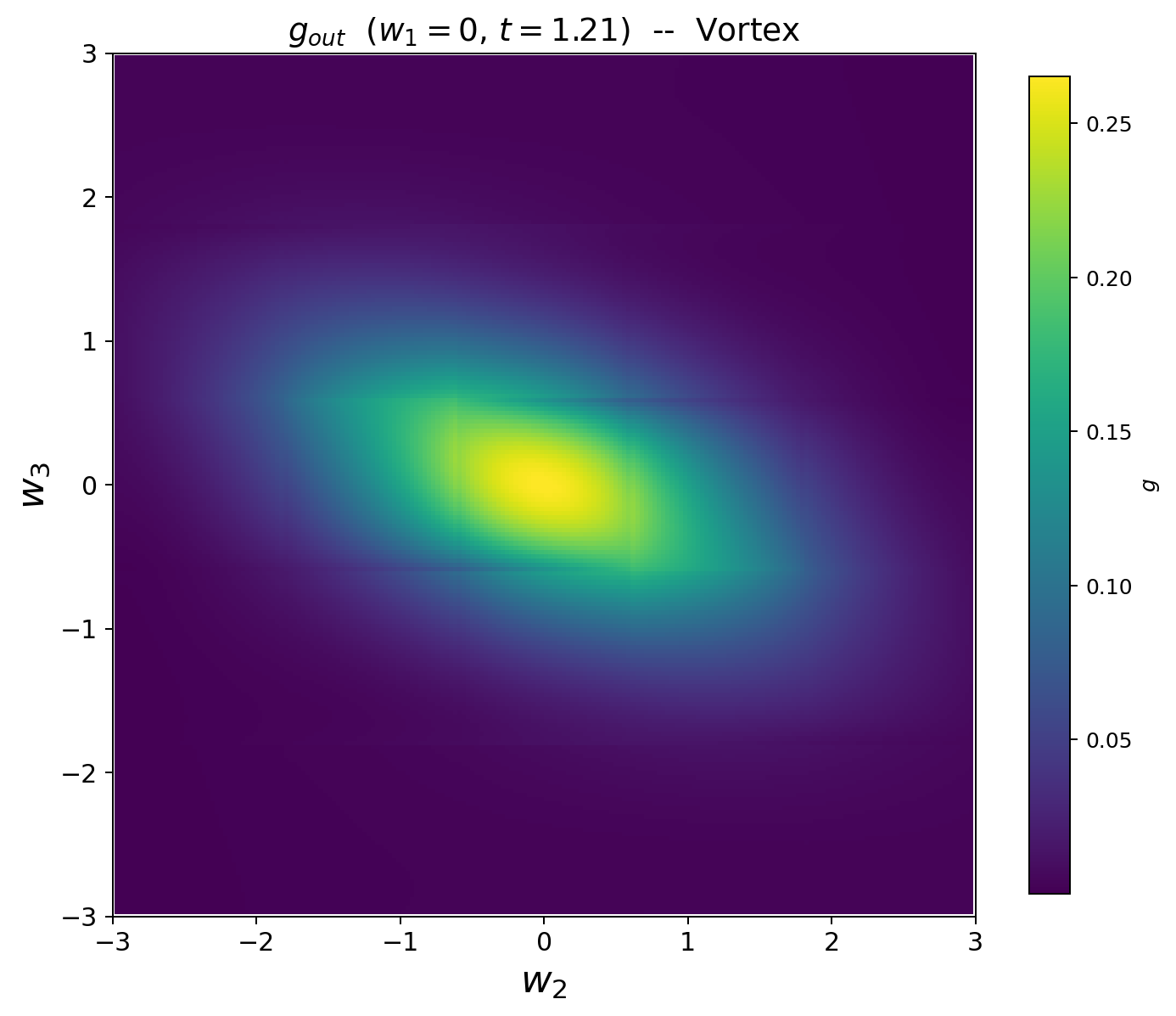}}
  \caption{Cross-sectional slice view of $g$ on the $3$ axes on the $5^3$ mesh for the vortex like flow.}
  \label{fig:g_at_1200}
\end{figure}

\begin{figure}[htbp]
  \centering
  \subfloat[Evolution of $g$ on $w_2=0,w_3=0$\label{fig:g_out_w1}]{\includegraphics[width=0.50\textwidth]{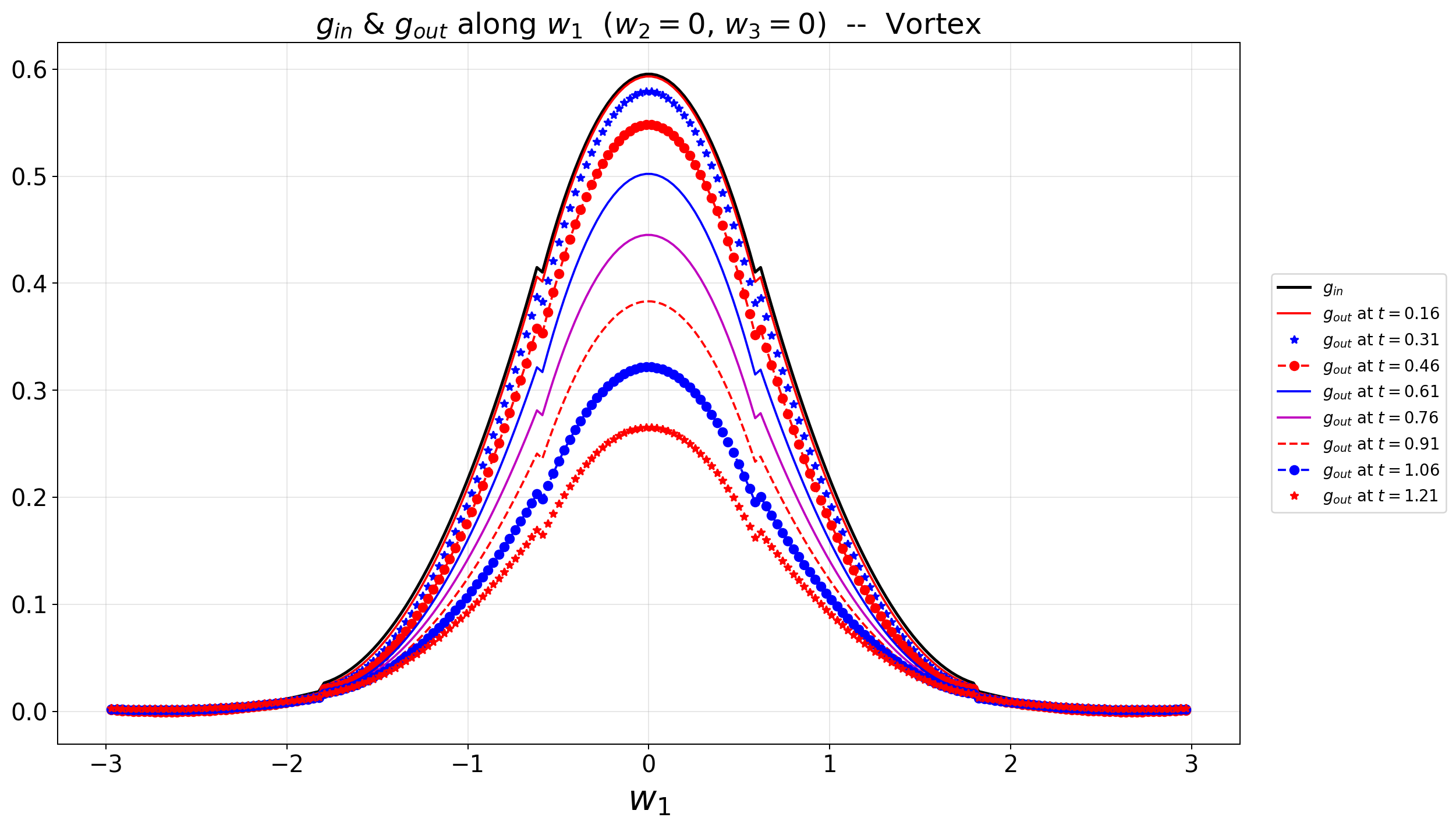}}\\[1ex]
  \subfloat[Evolution of $g$ on $w_1=0,w_2=0$\label{fig:g_out_w2}]{\includegraphics[width=0.50\textwidth]{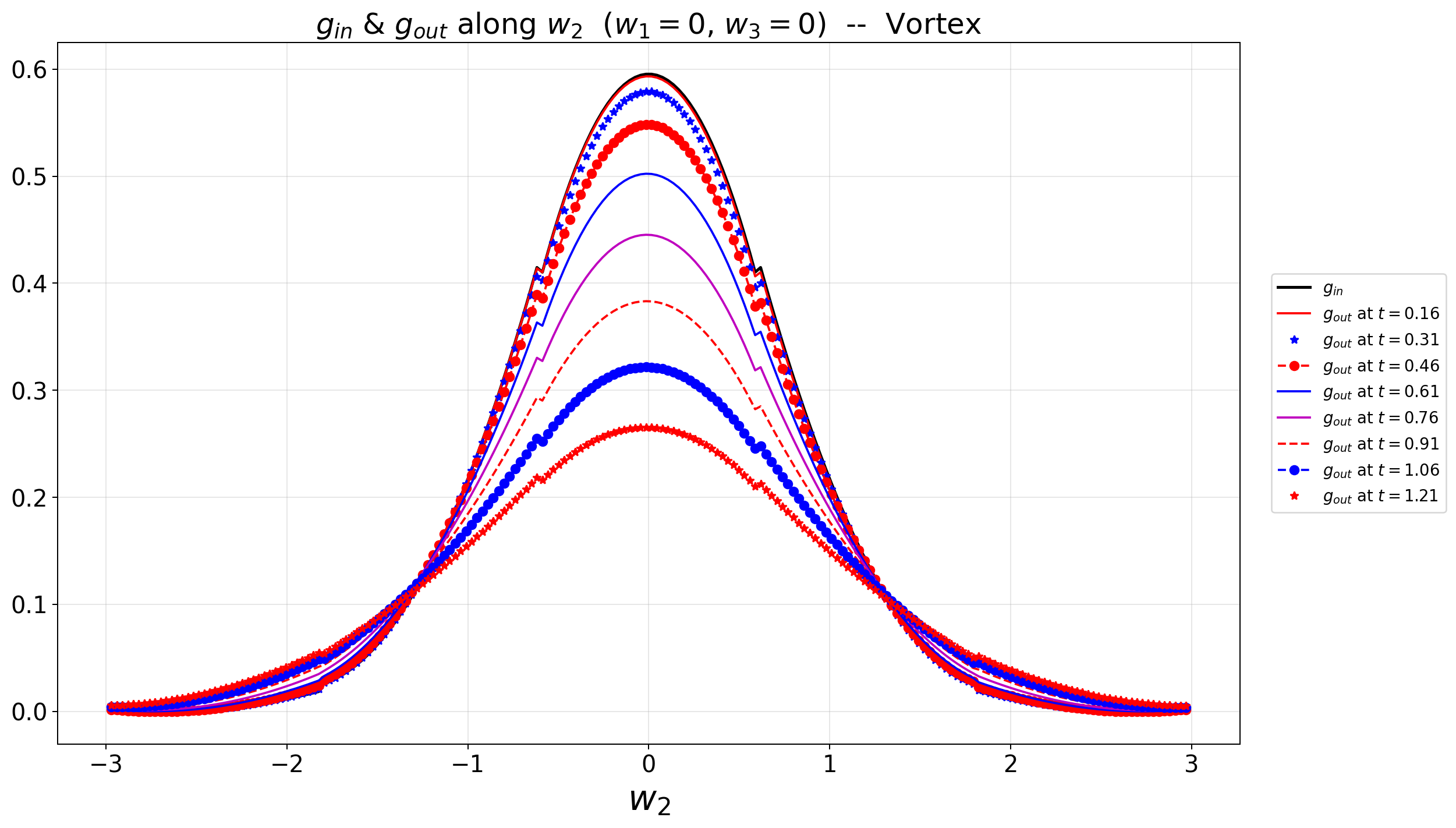}}\hfill
  \subfloat[Evolution of $g$ on $w_1=0,w_2=0$\label{fig:g_out_w3}]{\includegraphics[width=0.50\textwidth]{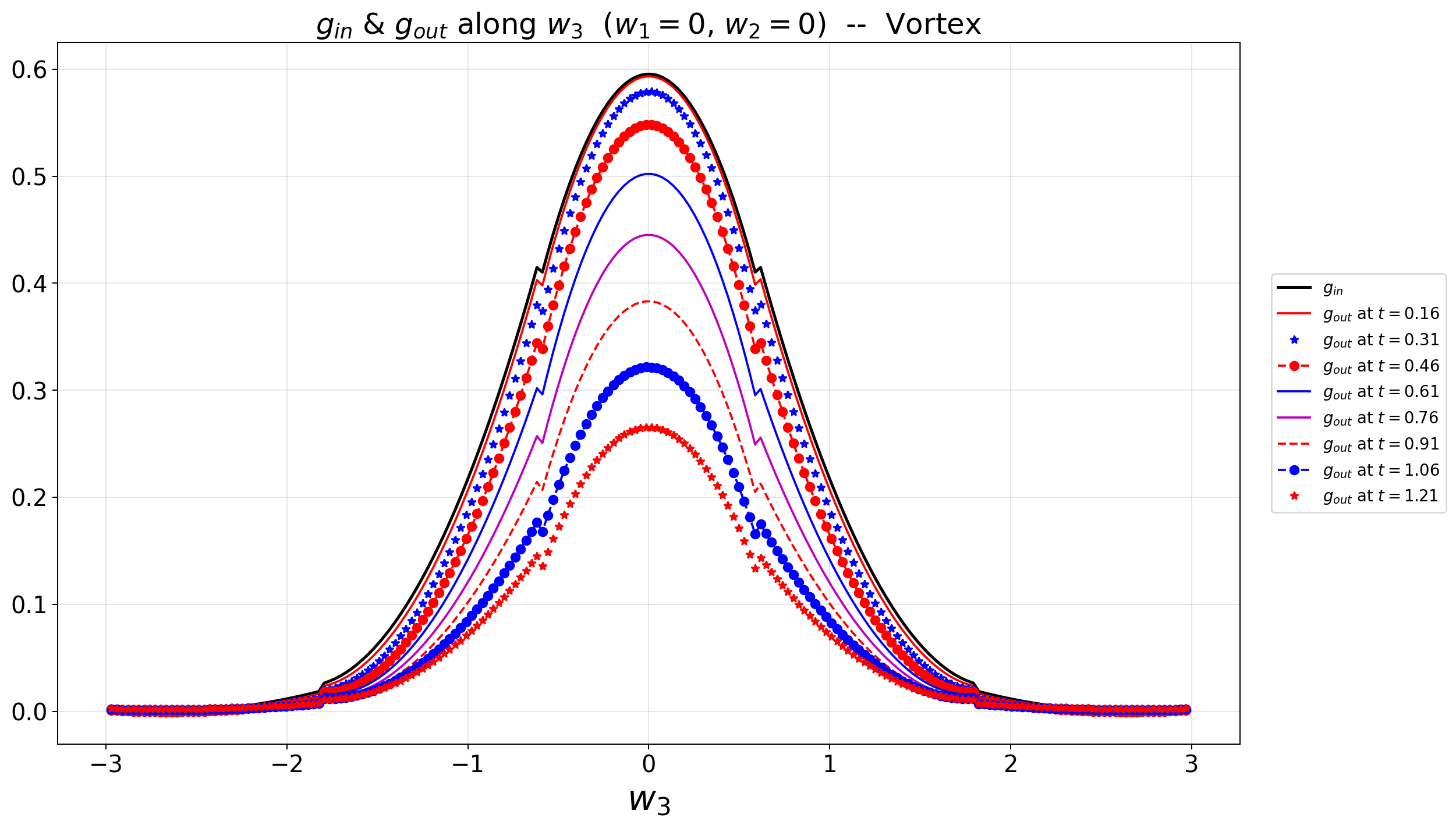}}
  \caption{Cross-sectional slice view of $g$ on the $3$ planes on the $5^3$ mesh for the vortex like flow.}
  \label{fig:g_evolution}
\end{figure}

We fit the velocity distribution at various times for all of the flows to an anisotropic (i.e., multivariate) Gaussian as described in \Cref{app:lm} using the $5^3$ (125-element) mesh.
\Cref{fig:residual} provides the relative residual $r = \|\gf_h - \gf_\text{fit}\|_{L_2}/\|\gf_h\|_{L_2}$ that quantifies the quality of the Gaussian approximation.
\Cref{fig:eigenvalues,fig:anisotropy,fig:angles} show the evolution of the fitted covariance through the eigenvalues $\lambda_1\geq\lambda_2\geq\lambda_3$ of $\bsigma$, the anisotropy ratio $\lambda_1/\lambda_3$, and the primary-plane principal axis angle, and are discussed in detail below.

\begin{figure}[htbp]
  \centering
  \includegraphics[width=0.6\textwidth]{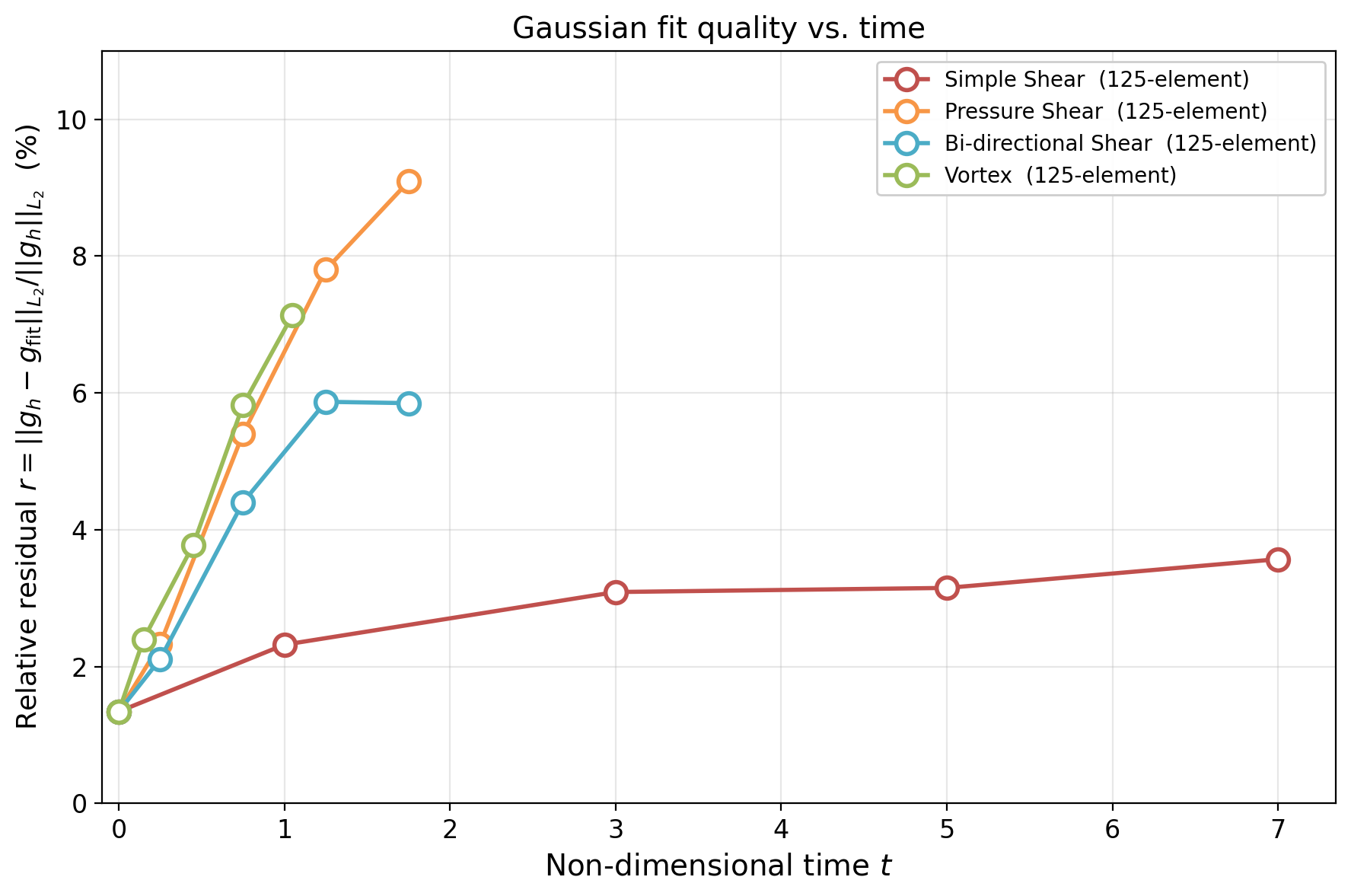}\hfill
  \caption{Relative $L_2$ residual of the Gaussian fit,
    $r = \|g_h - g_\text{fit}\|_{L_2} / \|g_h\|_{L_2}$, vs.\ time for the four affine flows on the $5^3$ mesh.}
  \label{fig:residual}
\end{figure}

\paragraph*{Simple shear.}
\label{sec:results_ss}

The largest eigenvalue $\lambda_1$ grows from $0.73$ to $3.53$ (factor
of $4.8$) over the simulation, while $\lambda_3$ increases by only a
factor of $2.0$.
The anisotropy ratio $\lambda_1/\lambda_3$ grows monotonically from
$1.91$ to $4.61$.
Off-diagonal covariance components $\Sigma_{13}$ and $\Sigma_{23}$ are
numerically negligible ($\lesssim 10^{-6}$), confirming that the flow
is purely two-dimensional in the $w_1$-$w_2$ plane, as expected from
$\bA$.

\paragraph*{Pressure shear.}

Pressure shear couples compression along $w_1$ ($A_{11}=-0.25$) with
shear in the $w_1$-$w_3$ plane ($A_{13}=1.4$).
The off-diagonal components $\Sigma_{12}$ and $\Sigma_{23}$ are negligible
($\lesssim 10^{-6}$), confirming a purely 2D flow in the $w_1$-$w_3$
plane.
The anisotropy ratio grows more rapidly than for simple shear
(reaching $6.84$ at $t=1.75$ vs.\ $4.61$ for simple shear at $t=7.0$
per unit time elapsed), consistent with the larger shear component
$A_{13}=1.4$.

\paragraph*{Bi-directional shear.}
Unlike simple shear and pressure shear, whose dynamics remain confined to a single reduced velocity-space plane ($w_1$-$w_2$ and $w_1$-$w_3$ respectively), bi-directional shear drives anisotropy in a fully
three-dimensional manner.
The three shear couplings $A_{12}=0.9$, $A_{13}=1.4$, and $A_{23}=0.7$ act simultaneously, so the principal eigenvector of $\bsigma(t)$ acquires non-negligible components along all three velocity axes, and all three planar marginals show significant principal-axis rotation (\Cref{fig:kinematic}).
This flow therefore exercises the full six-component covariance structure rather than a single 2$\times$2 sub-block, demonstrating that the anisotropic-Gaussian description and the free-streaming prediction $\bC^{-1}$ extend to genuinely three-dimensional anisotropy.

\paragraph*{Vortex.}

The vortex flow shows the fastest anisotropy growth among the four flows: $\lambda_1/\lambda_3$ reaches $9.95$ at $t=1.05$, whereas simple shear requires $t=7.0$ to reach $4.61$.
The eigenvalues $\lambda_2$ and $\lambda_3$ decrease initially over time (unlike simple shear, where all eigenvalues grow), indicating that the vortex preferentially concentrates the variance along a single direction and then eventually increases.

\subsection{Eigenvalue and Anistropy Evolution}
\label{sec:results_anisotropy}


\Cref{fig:eigenvalues} shows the eigenvalues of the fitted covariance matrix as functions of time for all four flows on both
meshes.
In all four flows, $\lambda_1$ grows significantly while $\lambda_2$ and $\lambda_3$ remain comparatively bounded, producing a characteristic one-directional stretching of the velocity distribution.
The two meshes agree closely on all eigenvalues; the $5^3$ mesh produces slightly larger $\lambda_1$ values (by $10$--$15\%$) due to its lower numerical diffusion.
A notable difference among the flows concerns the minor eigenvalues: for simple shear all three eigenvalues grow (the distribution stretches without compression), whereas for the vortex $\lambda_2$ and $\lambda_3$ decrease, reflecting the concentration of variance onto a single dominant direction by the rotational kinematics.

\begin{figure}[htb!]
  \centering
  \includegraphics[width=\textwidth]{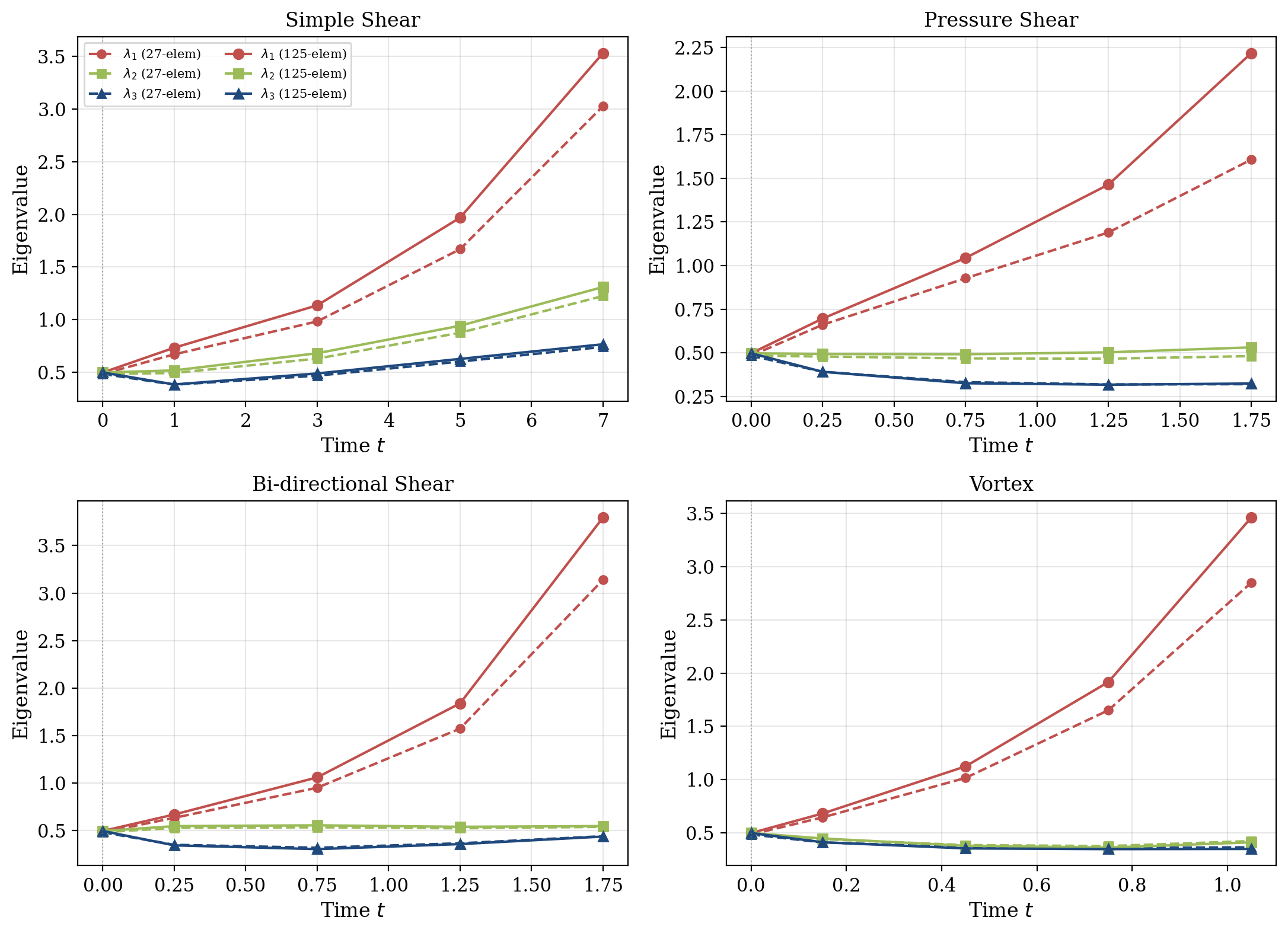}
  \caption{Eigenvalues $\lambda_1\geq\lambda_2\geq\lambda_3$ of the
    fitted covariance matrix $\bsigma(t)$ vs.\ time for all four flows.
    Solid lines: $5^3$ mesh. Dashed lines: $3^3$ mesh.
    The initial isotropic Maxwellian has
    $\lambda_1=\lambda_2=\lambda_3$.
    }
  \label{fig:eigenvalues}
\end{figure}


\Cref{fig:anisotropy} shows the anisotropy ratio
$\lambda_1/\lambda_3$ on a logarithmic scale for all four flows. The anisotropy ratio grows monotonically in all four flows, confirming
that the affine deformation continuously drives the distribution away
from isotropy. On the logarithmic scale the growth appears nearly linear in time for
the vortex and bi-directional shear, suggesting approximately
exponential anisotropy growth driven by the large spin components in
those flows.
Simple shear, by contrast, grows sub-exponentially: the nilpotent
structure of $\bA$ ($\bA^2=\bm{0}$) means $\bF(t)=\bI+t\bA$ grows
only linearly, so $\lambda_1\sim t^2$ at leading order rather than
$\sim e^{2\gamma t}$.
The final ratios, normalized by simulation duration, give anisotropy
growth rates of approximately $9.5$ per unit time for the vortex,
$4.9$ for bi-directional shear, $3.9$ for pressure shear, and $0.66$
for simple shear.
The vortex rate is roughly 19 times that of simple shear, consistent
with the much larger spin components of the vortex
($|W_{12}|=|W_{13}|=0.65$ vs.\ $|W_{12}|=0.40$) and the non-linear
growth enabled by a non-nilpotent $\bA$.

\begin{figure}[htb!]
  \centering
  \includegraphics[width=0.65\textwidth]{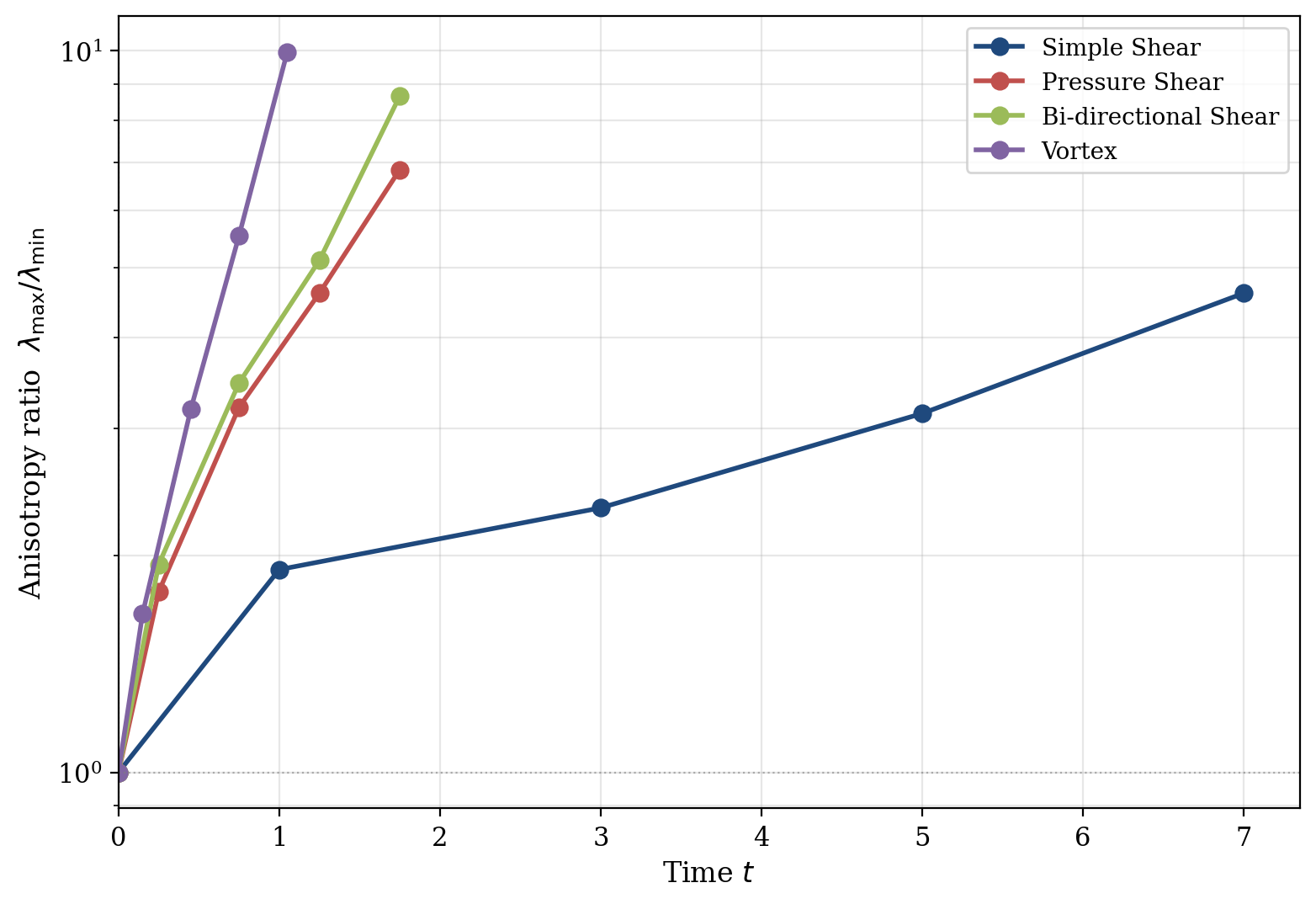}
  \caption{Anisotropy ratio $\lambda_1/\lambda_3$ vs.\ time on a
    logarithmic scale ($5^3$ mesh, all four flows).
    All flows show monotonic growth; the vortex reaches the largest
    ratio ($9.96$ at $t=1.05$) despite the shortest simulation time.}
  \label{fig:anisotropy}
\end{figure}

\subsection{Principal Axis Evolution}
\label{sec:results_angles}

\Cref{fig:angles} shows the evolution of the primary-plane principal axis angle, i.e., the (direction of the largest eigenvector of $\bsigma$  reported as a line direction in $[0^\circ,180^\circ)$), for
all four flows on both meshes.

\begin{figure}[htb!]
  \centering
  \includegraphics[width=\textwidth]{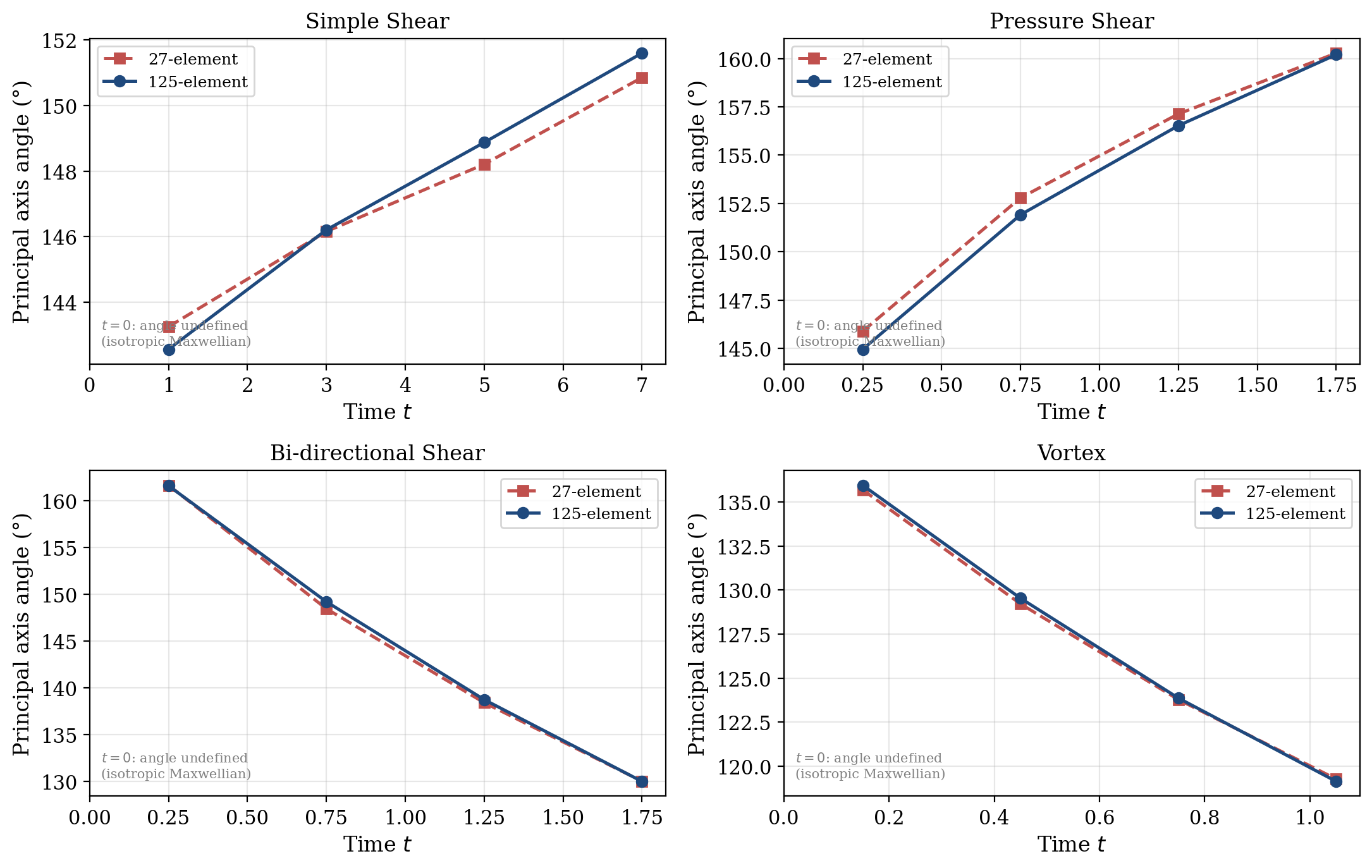}
  \caption{Principal axis angle in the primary velocity-space plane vs.\ time for all four flows.
    Solid lines: $5^3$ mesh. Dashed lines: $3^3$ mesh.
    Primary plane is $w_1$-$w_2$ for simple shear, bi-directional
    shear, and vortex; it is $w_1$-$w_3$ for pressure shear.
    The $3^3$ and $5^3$ meshes agree to within $\approx1^\circ$
    throughout.}
  \label{fig:angles}
\end{figure}

The rotation directions are flow-dependent and reflect the underlying kinematics of each flow.
For simple shear and pressure shear, the principal axis rotates counter-clockwise (angle increasing toward $180^\circ$).
Both flows have a single non-zero off-diagonal entry in $\bA$ in the primary plane ($A_{12}=0.8$ and $A_{13}=1.4$ respectively), so the strain rate $\bD$ has its extensional eigenvector at $135^\circ$ in that plane.
Starting from an isotropic initial condition, the distribution is immediately stretched toward $135^\circ$ and continues rotating counter-clockwise as the anisotropy grows.

For bi-directional shear the rotation is clockwise in the
$w_1$-$w_2$ plane (angle decreasing from $\approx162^\circ$),
despite the spin component $W_{12}=0.45>0$ which alone would drive
counter-clockwise rotation.
The reason is that the $w_1$-$w_3$ shear coupling $A_{13}=1.4$
is stronger than the $w_1$-$w_2$ coupling $A_{12}=0.9$, and the
full three-dimensional principal axis evolution projects as a net
clockwise rotation when viewed in the $w_1$-$w_2$ plane alone.
This illustrates that the rotation direction in any individual
plane cannot be inferred from the spin component in that plane
alone when the flow has genuinely three-dimensional structure.

For the vortex the rotation is also clockwise (angle decreasing
from $\approx134^\circ$), here consistent with the dominant spin
component $W_{12}=-0.65<0$ in the $w_1$-$w_2$ plane.

All four flows show strictly monotonic angle evolution throughout the simulation, confirming a  unidirectional rotation of the principal axis of anisotropy in the primary plane.
The $3^3$ and $5^3$ meshes agree to within $\approx1^\circ$ at every time step, confirming mesh convergence of the principal axis direction.

We next compare the per-plane principal axis angles $\theta_{ij}^{\text{num}}$, against the free-streaming prediction $\theta_{ij}^{\text{fs}}$, obtained from the two covariance matrices $\bsigma_{\text{fit}}(t)$ and $\bsigma_{\text{fs}}(t)$ respectively. These quantities are defined as follows.

The numerical covariance $\bsigma_{\text{fit}}(t)$ is the
covariance of the anisotropic Gaussian fitted to the FEM solution
$\gf_h(t)$ using the method explained in Appendix \ref{app:lm}. 
The free-streaming covariance is obtained in the collisionless limit:
\begin{equation}
  \bsigma_{\text{fs}}(t) \;=\; T_0\,\big(\bF(t)^{\mathsf{T}}\bF(t)\big)^{-1},
  \label{eq:sigma_fs}
\end{equation}
by transporting the isotropic initial covariance
$\bsigma(0)=T_0\bI$ under the affine deformation $\bF(t)=\bI+\bA t$ without considering collisions.

Applying \eqref{eq:principal_angle}, as explained in Appendix \ref{app:angles}, to each covariance gives the
numerical and free-streaming principal angles in the $(i,j)$ plane,
\begin{equation}
  \theta_{ij}^{\text{num}}(t) = \theta_{ij}\!\big(\bsigma_{\text{fit}}(t)\big),
  \qquad
  \theta_{ij}^{\text{fs}}(t)  = \theta_{ij}\!\big(\bsigma_{\text{fs}}(t)\big).
  \label{eq:two_angles}
\end{equation}
Both use only the three entries $\Sigma_{ii},\Sigma_{jj},\Sigma_{ij}$ of the relevant
subblock, so the comparison is planar and self-consistent. The scalar
$T_0$ in \cref{eq:sigma_fs} rescales $\bsigma_{\text{fs}}$ uniformly and
therefore does not affect $\theta_{ij}^{\text{fs}}$. The gap is defined as $\Delta\theta_{ij} = \theta_{ij}^\text{num} - \theta_{ij}^\text{fs}$, and measures the
rotation of the  principal axis induced by collisions,
relative to the free-streaming case.

\paragraph*{Simple shear: $w_1$-$w_2$ plane.}

Both the numerical data and the free-streaming prediction in \Cref{fig:kinematic} show counter-clockwise rotation from $\approx 143^\circ$ to $\approx 152^\circ$.
The free-streaming prediction starts in good agreement ($\Delta\theta=-3.3^\circ$ at $t=1$) and diverges to $-18.6^\circ$ at $t=7$, due to the collision-induced retardation accumulating over multiple collision times.

\paragraph*{Pressure shear: $w_1$-$w_3$ plane.}

Agreement at $t=0.25$ is excellent ($\Delta\theta=-0.2^\circ$) and the gap grows to $-8.6^\circ$ by $t=1.75$ (\Cref{fig:kinematic}).

\paragraph*{Bi-directional shear: all three planes.}

The $w_1$-$w_2$ marginal angle shows non-monotonic behaviour (139.2°$\allowbreak\to$142.5°$\allowbreak\to$137.1°$\allowbreak\to$129.7°), rising slightly at $t=0.75$ before decreasing; the free-streaming prediction captures this correctly, and the gap remains below $2.3^\circ$ throughout (\Cref{fig:kinematic}).
The $w_1$-$w_3$ plane shows a growing gap ($-8.6^\circ$) matching pressure shear at the same $t\cdot|W_{13}|=1.225$.
The $w_2$-$w_3$ plane passes through near-isotropy (axis flip between
$t=0.25$ and $t=0.75$); the numerical results track this and the gap is equal and opposite in sign to that of $w_1$-$w_3$.

\paragraph*{Vortex: all three planes.}

All three planes show sub-$2^\circ$
agreement at $t=0.45$ ($t\cdot|W_{ij}|\leq 0.29$).
By $t=1.05$ the gaps reach $+1.6^\circ$, $+4.5^\circ$, and $-5.1^\circ$.

\begin{figure}[htb!]
  \centering
  \includegraphics[width=\textwidth]{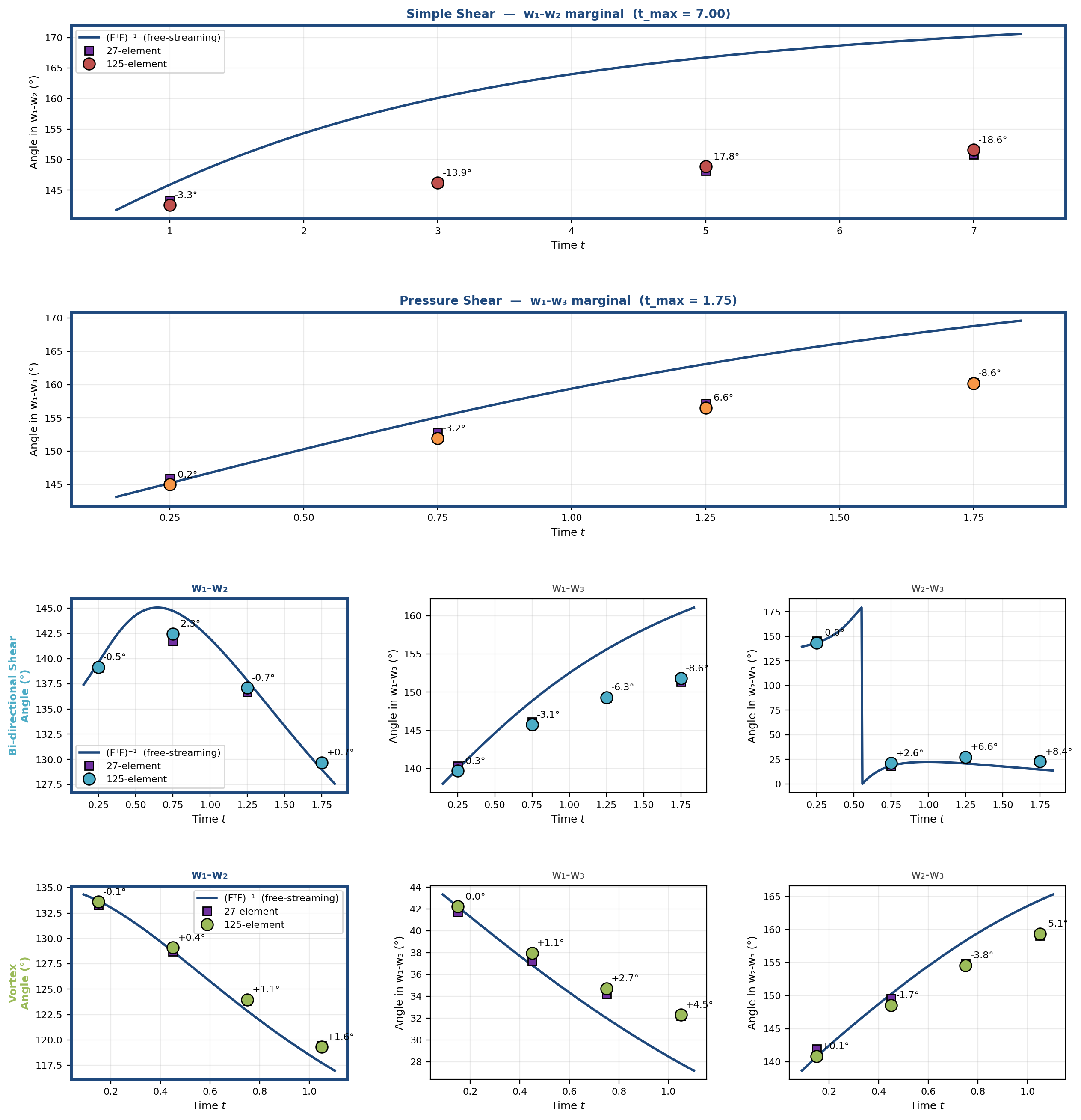}
  \caption{Per-plane principal axis angle vs.\ time for all four
    flows (detailed view complementing \Cref{fig:angles}).
    Solid curves: free-streaming prediction
    $\theta_{ij}^\text{fs}$ from \eqref{eq:two_angles}.
    Circles: $5^3$ data. Squares: $3^3$ data.
    Annotated numbers: gap $\Delta\theta_{ij}$ in degrees.
    Blue border: primary plane. $\Delta t=10^{-3}$.}
  \label{fig:kinematic}
\end{figure}
\section{Discussion and Conclusions}
\label{sec:discussion}

\subsection{Gaussian structure far from equilibrium}
\label{sec:gaussian_discussion}

The central finding of this work is that the velocity distribution
$\gf(\bw,t)$ remains well-approximated by an anisotropic Gaussian
throughout all simulated flows, even at strongly non-equilibrium
conditions.
The linear transport term in~\eqref{eq:reduced_bte} preserves Gaussian
structure exactly in the free-streaming limit: if $\gf_0$ is Gaussian
then $\gf_\text{fs}(\bw,t)=\gf_0(\bF(t)\bw)$ is also Gaussian, since
composition with a linear map preserves the Gaussian form.

The finding that $\gf$ remains well-approximated by an anisotropic Gaussian throughout all simulated flows has important physical implications.
The transport term drives it toward an anisotropic Gaussian in the free-streaming limit.
The collision operator, by the $H$-theorem, drives $\gf$ toward the isotropic Maxwellian.
Their combined nonlinear interaction keeps $\gf$ near an anisotropic Gaussian throughout.
For Maxwell molecules this near-Gaussian behavior can be understood theoretically via Bobylev's Fourier-mode solutions \cite{bobylev1976,bobylev1977}, which show that the Gaussian is exactly invariant under the Boltzmann operator for that kernel.
For hard spheres, used in the present work, there is no analogous result exists, but the near-Gaussian structure observed here suggests this possibility.
The quality of the approximation is quantified by the relative $L_2$ residual $r=\|\gf_h-\gf_\text{fit}\|_{L_2}/\|\gf_h\|_{L_2}$.
All 16 data points across the four flows remain below $10\%$:
Simple Shear $2.3$--$3.6\%$, Bi-directional Shear $2.1$--$5.9\%$,
Vortex $2.4$--$7.1\%$, and Pressure Shear $2.3$--$9.1\%$.
This accuracy is maintained at anisotropy ratios $\lambda_1/\lambda_3$
of up to $9.95$ (Vortex, $t=1.05$) and $8.66$ (Bi-directional Shear,
$t=1.75$).

Because we solve the full Boltzmann collision integral without any moment closure assumption, the Gaussian structure is discovered from the dynamics rather than imposed by modeling; a standard BGK approximation, for instance, drives $\gf$ toward an isotropic Maxwellian by construction, and would suppress or diminish the anisotropic covariance structure observed here.
Standard moment closures such as BGK assume a Maxwellian (single temperature) by construction; they cannot produce the anisotropic covariance structure or principal-axis rotation documented here.
As a secondary implication, the near-Gaussian structure suggests that a tensorial temperature with $\bsigma(t)$ as the state variable may be a useful starting point for developing reduced-order models of non-equilibrium flows. In kinetic theory the only stress is the thermal (kinetic) one,
proportional to the second moment of the reduced peculiar velocity, since momentum is transported only by particle motion and not through inter-particle forces; we therefore
refer to $\bsigma$ as the tensorial temperature, noting that it enters the energy balance~\eqref{eq:dedt} through the stress power because the thermal stress and the temperature tensor are the same object up to constant dimensional factors.

\subsection{Comparison with molecular dynamics}
\label{sec:omd_comparison}

As an additional validation, we compare the DG Boltzmann solution
for simple shear directly against non-equilibrium molecular dynamics
(NEMD) simulations carried out using LAMMPS \cite{Plimpton1995}.
The LAMMPS simulations model a low-density Argon gas in a cubic box
of side $2000\sigma$ (where $\sigma$ is the Lennard-Jones length
parameter for Ar) using $10{,}000$ and $100{,}000$ particles.
We approximate the hard-sphere potential by a truncated Lennard-Jones potential with no attraction.
The same dimensionless shear rate $\dot\gamma = 0.8$ is used in both simulations, with the mean free path expression
used to match the dimensionless parameters between the two methods.

The DG solution and the LAMMPS simulation are qualitatively
consistent: both show diagonal spreading of the distribution function
$\gf$ in the $w_1$-$w_2$ plane, confirming the development of
anisotropy in the direction of the applied shear \cite{debnath2018}.
The initial condition is symmetric with respect to all three axes
in both methods, and the evolved distributions show the same
orientation of the principal axis of anisotropy. We emphasise that this is a deterministic solution of the
Boltzmann equation: $g(\bw,t)$ is represented on a velocity-space
mesh and there are no simulation particles. The number density
$n=\int_{\mathbb{R}^3}g\,d\bw\approx3$ is a nondimensional quantity in
scaled velocity units and \textit{not} the number of particles.

\subsection{Pressure-Shear: dilatation vs.\ shear competition}
\label{sec:dilatative}

The pressure shear flow analysed in \Cref{sec:results} uses
$A_{11}=-0.25$ and $A_{13}=1.4$, in which the gas is compressed along
$w_1$ ($\operatorname{tr}\bA<0$) while being sheared in the $w_1$-$w_3$
plane.
This is the shear-dominated regime: the shear coupling
$A_{13}=1.4$ is roughly six times the magnitude of the compression
rate $A_{11}=-0.25$, and the resulting specific internal energy
$e(t) = \tfrac{1}{2}\operatorname{tr}\bsigma(t)$ grows monotonically
in agreement with the evolution equation
$\deriv{e}{t} = -\bD(t):\bsigma(t)$ derived in \Cref{sec:energy_discussion}.

Pahlani et al.\ \cite{pahlani2023a,pahlani2023b}, working in the
 molecular dynamics framework for the same class of
affine flows, reported a qualitatively different regime in which the
gas is expanded rather than compressed
($\operatorname{tr}\bA>0$, dilatation) while simultaneously sheared.
In this regime they observe a competition between two effects:
dilatation, which acts to lower the specific energy through expansion,
and shear, which raises the specific energy by developing
off-diagonal covariance components.
The competition can produce a non-monotonic energy trajectory with a
transient cooling phase followed by recovery once shear-induced
anisotropy is sufficient to overpower the dilatation cooling.

To investigate whether the present DG Boltzmann solver reproduces the
dilatation-vs-shear competition, we ran an additional simulation with
\begin{equation}
  \bA = \begin{pmatrix}
    0.3 & 0 & 1.2 \\
    0   & 0 & 0   \\
    0   & 0 & 0
  \end{pmatrix},
  \label{eq:dilatation_A}
\end{equation}
i.e.\ the same shear plane as the compressive pressure shear flow but
with the diagonal entry $A_{11}$ flipped in sign and the shear
component slightly reduced ($1.4\to1.2$).
This choice produces $\operatorname{tr}\bA=+0.3$, so the flow is
expansive: the deformation gradient determinant is
$\det\bF(t) = 1+0.3t$, and the analytically expected number density
is $n(t)=n_0/(1+0.3t)$.
For this case the conservation routine of \Cref{sec:numerics} was
modified to target the time-dependent value $n_0/\det\bF(t)$ rather
than the constant $n_0$, as described in \Cref{sec:conservation_fix},
making the constraint consistent with the 
compressibility of the flow.

\Cref{fig:dilatative_n} shows the simulated number density compared
to the analytical prediction $n_0/(1+0.3t)$.
The FEM solution tracks the prediction to within numerical precision
throughout the simulated time range, confirming that the modified
conservation routine correctly enforces the physically expected
compressibility.

\begin{figure}[htb!]
  \centering
  \includegraphics[width=0.45\textwidth]{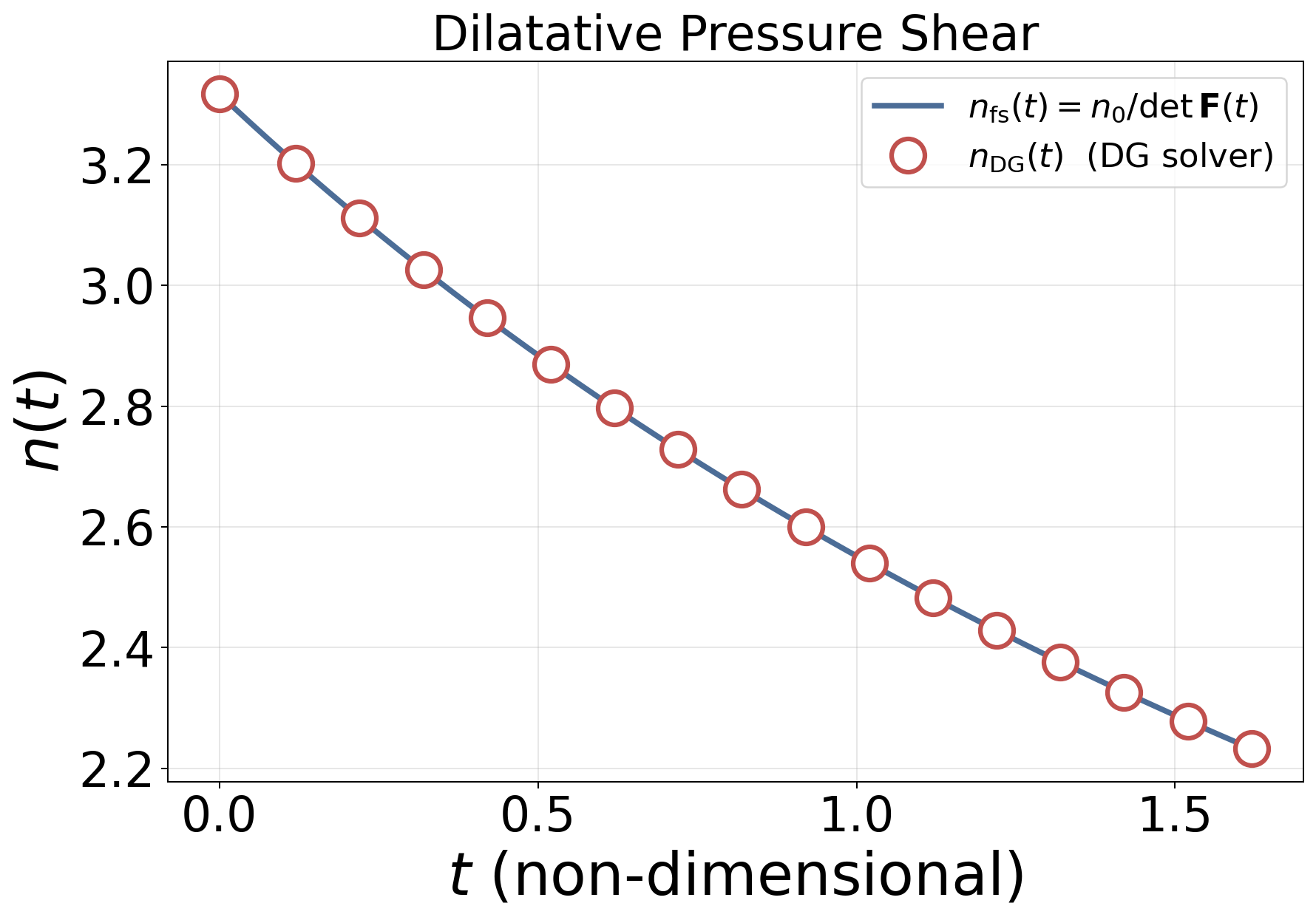}
  \caption{Number density $n(t)$ for the dilatative pressure shear
    flow~\eqref{eq:dilatation_A}.
    The FEM solution (markers) tracks the analytical prediction
    $n_0/(1+0.3t)$ (solid line) throughout, demonstrating that the
    discrete solver correctly captures the compressibility of the
    flow.}
  \label{fig:dilatative_n}
\end{figure}

\Cref{fig:dilatative_e} shows the specific internal energy
$e(t) = \tfrac{1}{2}\operatorname{tr}\bsigma(t)$ computed three ways:
the free-streaming prediction
$e_\text{fs}(t) = \tfrac{T_0}{2}\operatorname{tr}\!\bigl((\bF^\T\bF)^{-1}\bigr)$,
the Gaussian-fitted DG energy
$e_\text{fit}(t) = \tfrac{1}{2}\operatorname{tr}\bsigma_\text{fit}(t)$,
and the direct moment of the DG solution
$e_\text{DG}(t)$.
All three curves exhibit a clear U-shape: the energy decreases
from $e(0)\approx 0.746$ to a minimum near $t^* \approx 0.30$, then
recovers as shear-induced off-diagonal covariance grows.
This is precisely the dilatation-vs-shear competition described by
Pahlani et al.\ \cite{pahlani2023b}.

A quantitative gap develops between $e_\text{fit}$ and $e_\text{DG}$ at later times, reflecting the velocity-domain truncation in the direct moment: the truncated $[-3,3]^3$ domain loses contributions from the tails of the stretched distribution, while the Gaussian fit extrapolates beyond the domain to recover the full second moment.

\begin{figure}[htbp]
  \centering
  \includegraphics[width=0.7\textwidth]{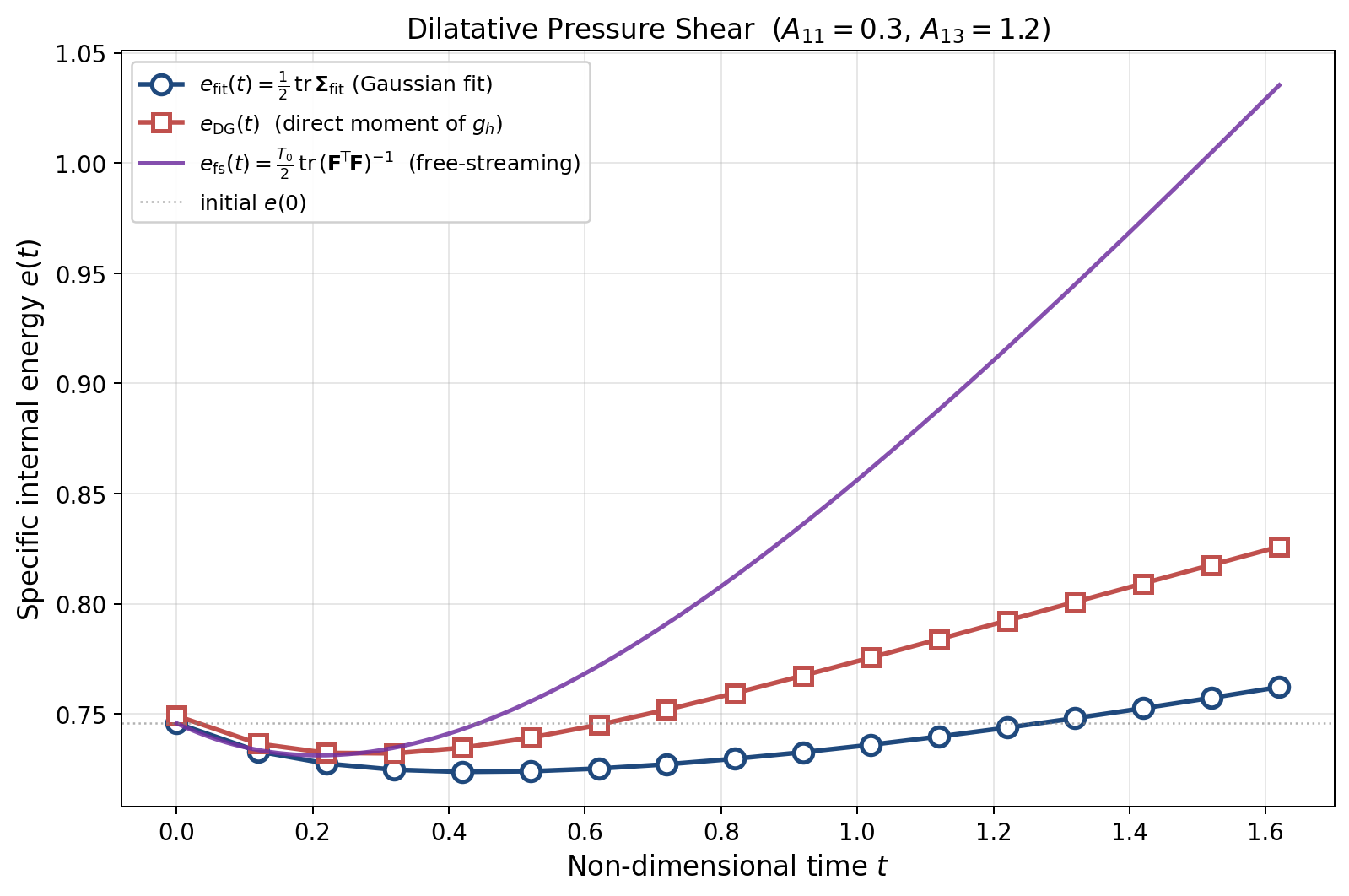}
  \caption{Specific internal energy $e(t)$ for the dilatative pressure
    shear flow~\eqref{eq:dilatation_A}.
    Solid line: free-streaming prediction
    $e_\text{fs}(t)=\tfrac{T_0}{2}\operatorname{tr}\!\bigl((\bF^\T\bF)^{-1}\bigr)$.
    Open circles: Gaussian-fitted DG covariance
    $e_\text{fit}(t)=\tfrac{1}{2}\operatorname{tr}\bsigma_\text{fit}(t)$.
    Open squares: direct DG moment
    $e_\text{DG}(t)=\tfrac{1}{2}\langle|\bw|^2\rangle_h$.
    All three curves show the U-shape characteristic of the
    dilatation-vs-shear competition first observed by Pahlani et al.\
    \cite{pahlani2023b} in MD simulations.}
  \label{fig:dilatative_e}
\end{figure}

The reduced energy evolution equation
$\deriv{e}{t} = -\bD(t):\bsigma(t)$ admits a natural decomposition into two
physically distinct contributions:
\begin{equation}
    \deriv{e}{t}
  = \dot{e}_\text{dil}(t) + \dot{e}_\text{shr}(t)
  = \underbrace{-\bD(t):(\bsigma(t)\circ\bI)}_{\text{diagonal (dilatation/compression)}}
  \;\underbrace{-\bD(t):(\bsigma(t)-\bsigma(t)\circ\bI)}_{\text{off-diagonal (shear)}},
  \label{eq:dedt_decomp}
\end{equation}
where $\circ$ denotes the Hadamard product and $\bsigma\circ\bI$
extracts the diagonal part of $\bsigma$.
The dilatational term $\dot{e}_\text{dil}$ contracts the strain rate
with the diagonal of the covariance and is non-zero from $t=0$ since
$\bsigma(0)=T_0\bI$; the shear term $\dot{e}_\text{shr}$ contracts
$\bD$ with the off-diagonal components of $\bsigma$ and is zero
initially.

At $t=0$, the diagonal contribution is
$-T_0\operatorname{tr}\bD(0) = -0.3\,T_0 \approx -0.149$, so
$\left.\deriv{e}{t}\right|_0 < 0$ and dilatation cooling dominates.
As $\Sigma_{13}$ develops a negative value under the positive shear
$A_{13}>0$, the off-diagonal contribution $-2D_{13}\Sigma_{13}$ grows
positive (since $D_{13}>0$ and $\Sigma_{13}<0$), overcoming
the dilatation cooling.
The crossover at $t^* \approx 0.30$ marks the point at which the two
contributions balance; beyond this time the shear heating dominates
and the energy increases.
The dilatation-vs-shear competition in pressure shear was previously known only from particle-based OMD simulations.
The present DG solution reproduces this physics deterministically and provides a quantitative decomposition into the two competing contributions via~\eqref{eq:dedt_decomp}.
The agreement of the U-shape between $e_\text{fs}$ and $e_\text{fit}$ further shows that collisions reduce the late-time recovery rate but do not modify the existence of the energy minimum, demonstrating that the free-streaming mechanism remains the dominant feature of the energy evolution even under collisional dynamics.

\subsection{Transient versus long-time behaviour}
The anisotropic Gaussian behaviour characterized here describes the
transient regime. 
For hard potentials, rigorous results establish that in the long-time, collision-dominated regime the distribution approaches an isotropic Maxwellian with growing temperature \cite{james2019b,kepka2024,duanliu2025}.
Our results are consistent with this picture: the isotropizing action of collisions is visible within the simulated window as the lag of the fitted principal-axis angle behind the free-streaming  prediction (\Cref{fig:angles}) and as the departure of the specific energy from the free-streaming value (\Cref{fig:dilatative_e}), both of which vanish identically in the collisionless limit, where the exact covariance is $T_0 (\bF^\T\bF)^{-1}$. 
Reaching the collision-dominated regime numerically, with its unbounded temperature growth, is challenging since the fixed velocity domain bounds the reliable simulation window; a quantitative comparison of the collision and transport terms, and of the covariance growth rates against the power laws of
\cite{james2019b,kepka2024,duanliu2025}, are important future directions.

We note finally that the unbounded temperature growth driving the asymptotic isotropization is specific to the idealized homoenergetic setting: in physical realizations, heat conduction (or a thermostat in molecular simulation) stabilizes the temperature, and steadily sheared states remain anisotropic
\cite{garzo2003}. The transient anisotropic regime characterized here is therefore relevant for numerous physically realistic situations wherein the temperature reaches a steady state.

\section*{}

\paragraph*{Software and Data Availability.}

The DG solver code, the fitted-Gaussian results for all four flows, and the post-processing scripts are available at \url{https://github.com/debnath2109/DG-Boltzmann-for-affine-flows}.
The raw simulation input files (\texttt{G\_INIT.dat} as initial condition), precomputed Collision data, and the output files (\texttt{gout.dat}) are archived at \href{https://doi.org/10.5281/zenodo.20518886}
{\texttt{doi.org/10.5281/zenodo.20518886}}.

\paragraph*{Acknowledgments.}

We acknowledge Richard D. James for insightful comments; AFOSR (MURI FA9550-18-1-0095) and ARO (MURI W911NF-24-2-0184) for financial support; Air Force Research Laboratory for hosting Kaushik Dayal's research visits; and NSF ACCESS for computing resources provided by Pittsburgh Supercomputing Center.


\appendix

\makeatletter
\renewcommand*{\thesection}{\Alph{section}}
\renewcommand*{\thesubsection}{\thesection.\arabic{subsection}}
\renewcommand*{\p@subsection}{}
\renewcommand*{\thesubsubsection}{\thesubsection.\arabic{subsubsection}}
\renewcommand*{\p@subsubsection}{}
\makeatother

\section{Reconstruction Algorithm}
\label{app:reconstruction}

Given the DG nodal values $\{g_{\bm\alpha}^{(e)}\}$, the reconstruction
of $\gf_h$ on a $200^3$ uniform grid proceeds as follows.

\begin{algorithm}[H]
\caption{DG reconstruction on $200^3$ grid}
\label{alg:reconstruction}
\begin{algorithmic}[1]
\State Set $c_\text{max}=3$, $N_\text{grid}=200$
\State Allocate $G[200,200,200]$ initialised to zero
\For{each element $e = 1,\ldots,N_e^3$}
  \State Determine the range of grid indices $[i_1,i_2]\times[j_1,j_2]\times[k_1,k_2]$
         covered by element $e$
  \For{each grid point $(w_1^i,w_2^j,w_3^k)$ in the element range}
    \State Map to reference coordinates:
           $\xi_m = 2(w_m - w_m^{(e)})/h$, $m=1,2,3$
    \State Evaluate $G[i,j,k] = \sum_{\bm\alpha}
           g_{\bm\alpha}^{(e)}\,\ell_{\alpha_1}(\xi_1)
           \ell_{\alpha_2}(\xi_2)\ell_{\alpha_3}(\xi_3)$
  \EndFor
\EndFor
\State \textbf{return} $G$
\end{algorithmic}
\end{algorithm}

The reconstruction is exact (no quadrature error) at the $200^3$
grid points, since the DG polynomial is evaluated directly via the
Lagrange formula.

\section{Best-fit of the Reduced Velocity Distribution}
\label{app:lm}

At selected time steps, we fit the reduced velocity distribution $\gf_h$ to the anisotropic Gaussian form:
\begin{equation}
  \gf_\text{fit}(\bw) = A_0
    \exp\!\bigl(-\tfrac{1}{2}\,\bw^\T\bsigma^{-1}\bw\bigr),
  \label{eq:Gfit}
\end{equation}
with seven free parameters $\theta=(\Sigma_{11},\Sigma_{12},\Sigma_{13},\Sigma_{22},\Sigma_{23},\Sigma_{33},A_0)$: the 6 independent components of the symmetric positive-definite covariance matrix $\bsigma$ and the amplitude $A_0$.
The Levenberg--Marquardt algorithm \cite{levenberg1944,marquardt1963} is used to solve the nonlinear least-squares problem:
\begin{equation}
  \min_\theta \; R(\theta) =
  \bigl\|\gf_h - \gf_\text{fit}(\,\cdot\,;\theta)\bigr\|^2_{\ell^2(\Omega_h)},
  \label{eq:resnorm_def}
\end{equation}
where $\Omega_h$ is the $200^3$ grid.

To ensure positive definiteness of $\bsigma(\theta)$ throughout the iteration, the covariance is parameterized via its Cholesky factorization $\bsigma = \bL_\Sigma\bL_\Sigma^\T$, with the six lower-triangular entries of $\bL_\Sigma$ as the free parameters replacing the six components of $\bsigma$.

The initial guess is $\bsigma=\bI$ (isotropic), $A_0=1/(2\pi)^{3/2}$, and the solution from the previous time step is used as a warm start for subsequent snapshots.
Convergence is declared when the relative change in $R$ between Levenberg--Marquardt iterations falls below $10^{-8}$.

\section{In-plane Principal Angle.}
\label{app:angles}

Let $\bsigma$ be a symmetric positive-definite $3\times3$ covariance
matrix and let $(i,j)$ be a coordinate plane, $(i,j)\in
\{(1,2),(1,3),(2,3)\}$. The marginal distribution of $(w_i,w_j)$,
obtained by integrating the Gaussian over the third velocity
component, is a planar Gaussian whose covariance is the $2\times2$
subblock
\begin{equation}
  \bsigma^{(ij)} =
  \begin{pmatrix} \Sigma_{ii} & \Sigma_{ij}\\[2pt] \Sigma_{ij} & \Sigma_{jj} \end{pmatrix}.
  \label{eq:subblock}
\end{equation}
Its iso-probability contours are ellipses, and the major axis of the
ellipse is the in-plane direction of largest variance, i.e.\ the
principal direction of the symmetric tensor $\bsigma^{(ij)}$. Let
$\bm{e}_1(\bsigma^{(ij)})=(v_i,v_j)$ be the eigenvector associated with
the larger eigenvalue of $\bsigma^{(ij)}$, with components $v_i$ and $v_j$
along the $w_i$- and $w_j$-axes. Measured from the $w_i$-axis
toward the $w_j$-axis, and regarded as an undirected line, this
direction makes the angle
\begin{equation}
  \theta_{ij}(\bsigma) =
  \tan^{-1}(v_j / v_i), \quad \theta_{ij} \in [0^\circ,180^\circ),
  \label{eq:principal_angle}
\end{equation}


\bibliographystyle{alpha}
\bibliography{boltzmann}

\end{document}